\documentclass{amsart}

\usepackage{amsmath}
\usepackage{mathtools}
\usepackage{amssymb}
\usepackage{enumerate}
\usepackage{slashed}
\usepackage{graphicx}
\usepackage{newlfont}
\usepackage{amsrefs}
\usepackage{comment}
\usepackage[normalem]{ulem}
\usepackage{mathtools}
\usepackage{accents}
\usepackage{leftidx}

\numberwithin{equation}{section}

\theoremstyle{plain}
\newtheorem{theorem}{Theorem}
\newtheorem{lemma}[theorem]{Lemma}

\newtheorem{proposition}[theorem]{Proposition}

\theoremstyle{definition}

\theoremstyle{remark}
\newtheorem{remark}{Remark}

\renewcommand{\div}{\operatorname{div}}

\newcommand{\R}{\mathbb{R}}

\newcommand{\D}{\mathcal{D}}

\renewcommand{\vector}[2]{\begin{pmatrix} #1 \\ #2 \end{pmatrix}}

\newcommand{\dist}{\operatorname{dist}}

\begin{document}

\title{Stationary Vacuum Black Holes in 5 Dimensions}

\author[Khuri]{Marcus Khuri}
\address{Department of Mathematics\\
Stony Brook University\\
Stony Brook, NY 11794, USA}
\email{khuri@math.sunysb.edu}

\author[Weinstein]{Gilbert Weinstein}
\address{Physics Department and Department of Mathematics\\
Ariel University\\
Ariel, 40700, Israel}
\email{gilbertw@ariel.ac.il}

\author[Yamada]{Sumio Yamada}
\address{Department of Mathematics\\
Gakushuin University\\
Tokyo 171-8588, Japan}
\email{yamada@math.gakushuin.ac.jp}

\thanks{M. Khuri acknowledges the support of
NSF Grants DMS-1308753 and DMS-1708798. S. Yamada acknowledges the support of JSPS grants KAKENHI 24340009 and 17H01091.}

\begin{abstract}
We study the problem of asymptotically flat bi-axially symmetric stationary solutions of the vacuum Einstein equations in $5$-dimensional spacetime. In this setting, the cross section of any connected component of the event horizon is a prime $3$-manifold of positive Yamabe type, namely the $3$-sphere $S^3$, the ring $S^1\times S^2$, or the lens space $L(p,q)$. The Einstein vacuum equations reduce to an axially symmetric harmonic map with prescribed singularities from $\R^3$ into the symmetric space $SL(3,\R)/SO(3)$. In this paper, we solve the problem for all possible topologies, and in particular the first candidates for smooth vacuum non-degenerate black lenses are produced. In addition, a generalization of this result is given in which the spacetime is allowed to have orbifold singularities. We also formulate conditions for the absence of conical singularities which guarantee a physically relevant solution.
\end{abstract}
\maketitle

\section{Introduction}
\label{sec1} \setcounter{equation}{0}
\setcounter{section}{1}

A result of Hawking \cite{Hawking} shows that a cross section of any connected component of the event horizon in a $4$-dimensional asymptotically flat stationary spacetime satisfying the dominated energy condition, has positive Euler characteristic, and hence must be topologically a $2$-sphere. The conclusion also holds without the stationarity condition provided one replaces a cross section of the event horizon with a stable apparent horizon. These results were generalized by Galloway and Schoen \cite{GallowaySchoen} to show that a cross section of any connected component of the event horizon in an $n$-dimensional asymptotically flat stationary spacetime is an $(n-2)$-dimensional Riemannian manifold with positive Yamabe invariant. In dimension $5$ the additional hypothesis of bi-axial symmetry restricts the possible topologies further, so that the only admissible topologies are $S^3$, $S^1\times S^2$, and $L(p,q)$ \cite{HollandsYazadjiev}. Explicit examples of stationary vacuum bi-axisymmetric solutions
with horizon topology $S^3$ and $S^1\times S^2$ have been constructed by Myers-Perry (sphere) \cite{MyersPerry},
Emparan-Reall (singly spinning ring) \cite{EmparanReall}, and Pomeransky-Sen'kov (doubly spinning ring) \cite{PomeranskySenkov}.
In particular, stationary vacuum black holes are not determined solely by their mass and angular momenta in higher dimensions. That is, the no-hair conjecture fails, as
there exist black ring solutions having the same
mass and angular momenta as a Myers-Perry black hole.
Nonetheless the underlying result supporting the validity of the no-hair theorem in 4-dimensions, a uniqueness theorem for harmonic maps with prescribed singularities into a nonpositively curved target, still holds in higher dimensions. In particular any bi-axially symmetric stationary vacuum solution is determined by a finite set of parameters. It is the purpose of this paper to establish a partial converse: given any admissible set of parameters, there is a unique solution of the reduced equations. Whether this solution of the reduced equations then generates a physical spacetime solution then depends on the absence of conical singularities on the axes.

The axes correspond to the locus where a closed-orbit Killing field degenerates, and in the domain $\R^3$ of the harmonic map these are identified by a number of intervals on the $z$-axis called \emph{axis rods}. The axis rods are separated by intervals corresponding to horizons, and by points which are referred to as \emph{corners}. Note that this precludes the case of degenerate horizons, in which horizons are represented by points instead of intervals. In addition, the end points of the horizon rods are named \emph{poles}. Denote by $\Gamma$ the $z$-axis with the interior of all the horizon rods removed, and let $\{p_l\}$ represent the corners and poles. Note that there are always two semi-infinite axes, labeled north and south. We assign a pair of relatively prime integers $(m_l,n_l)$ called the \emph{rod structure} to each axis rod $\Gamma_l$, such that the north and south semi-infinite axes are assigned the rod structures $(1,0)$ and $(0,1)$, respectively. This pair of numbers indicates which linear combination of rotational Killing fields vanishes on the associated rod.
If $(m_l,n_l)$ and $(m_{l+1},n_{l+1})$ are the rod structures assigned to two consecutive axis rods separated by a corner, then the \emph{admissibility condition} \cite{HollandsYazadjiev} is
\begin{equation} \label{admissibility0}
\det\begin{pmatrix} m_l & n_l \\ m_{l+1} & n_{l+1} \end{pmatrix} = \pm 1.
\end{equation}
This condition is to prevent orbifold singularities at the corners \cite{Evslin}.
Horizon rods are assigned the rod structure $(0,0)$. Finally, assign to each axis rod $\Gamma_l$ a constant $\mathbf{c}_l\in\R^2$, the \emph{potential constant}. The difference between the values of these constants on two axes adjoining a horizon rod is proportional to the angular momenta of this horizon component, as calculated by Komar integrals. A \textit{rod data set} $\D$ consists of the corners and poles $\{p_l\}$, the rod structures $\{(m_l,n_l)\}$, and the potential constants $\{\mathbf{c}_l\}$ which are assumed not to vary between two consecutive rods separated by a corner. This data determines uniquely the prescribed singularities of the harmonic map $\varphi\colon\R^3\setminus\Gamma\to SL(3,\R)/SO(3)$ as described more precisely in the Section \ref{sec4}, and will be referred to as admissible if it satisfies \eqref{admissibility0} at each corner. For technical reasons an additional \textit{compatibility condition} will be imposed to aid the existence result. This condition only applies when two consecutive corners are present. Let $p_{l-1}$ and $p_{l}$ be two consecutive corners with axis rods $\Gamma_{l-1}$ above $p_{l-1}$, $\Gamma_{l}$ between $p_{l-1}$ and $p_{l}$, and $\Gamma_{l+1}$ below $p_{l}$.
Then the compatibility condition states that the first component of the rod structures for $\Gamma_{l-1}$ and $\Gamma_{l+1}$ have opposite sign if both are nonzero
\begin{equation}\label{compatibilitycondition}
m_{l-1}m_{l+1}\leq 0,
\end{equation}
whenever the determinants \eqref{admissibility0} for the two corners $p_{l-1}$ and $p_{l}$ are both $+1$. Note that this latter requirement on the determinants may always be achieved by multiplying each component of the rod structures for $\Gamma_{l-1}$ and $\Gamma_{l}$ by $-1$ if necessary; this is an operation which does not change the properties of a rod.

In order to determine the physical relevance of a solution,
define on each bounded axis rod $\Gamma_l$ a function
$b_l$ to be the logarithm of the limiting ratio between the length of the closed orbit of the Killing field degenerating on $\Gamma_l$, and $2\pi$ times the radius from $\Gamma_l$ to this orbit. It turns out that $b_l$ is constant on $\Gamma_l$. The absence of a conical singularity on $\Gamma_l$ is the \emph{balancing condition} $b_l=0$.

An asymptotically flat stationary vacuum spacetime will be referred to as \textit{well-behaved} if the
orbits of the stationary Killing field are complete, the
domain of outer communication (DOC) is globally hyperbolic, and the DOC
contains an acausal spacelike connected hypersurface which is asymptotic to the canonical slice in the asymptotic end and whose boundary
is a compact cross section of the horizon. These assumptions are consistent with those of \cite{ChruscielCosta}, and are used for the reduction of the stationary vacuum equations. The main result may now be stated as follows.

\begin{theorem} \label{main}\par
\noindent
\begin{enumerate}[(i)]
\item
A well-behaved 5-dimensional asymptotically flat, stationary, bi-axially symmetric solution of the vacuum Einstein equations without degenerate horizons gives rise to a harmonic map $\varphi\colon\R^3\setminus\Gamma\to SL(3,\R)/SO(3)$ with prescribed singularities associated with an admissible rod data set $\D$, and satisfying $b_l=0$ on all bounded axis rods.
\item
Conversely,
given an admissible rod data set $\D$ satisfying the compatibility condition \eqref{compatibilitycondition}, there is a unique harmonic map $\varphi\colon\R^3\setminus\Gamma\to SL(3,\R)/SO(3)$ with prescribed singularities on $\Gamma$ corresponding to $\D$.
\item
A well-behaved 5-dimensional asymptotically flat, stationary, bi-axially symmetric solution of the vacuum Einstein equations without degenerate horizons can be constructed from $\varphi$ if and only if the resulting metric coefficients are sufficiently smooth across $\Gamma$ and $b_l=0$ on any bounded axis rod.
\end{enumerate}
\end{theorem}

The reduction of the Einstein vacuum equations to a harmonic map is well known \cites{Harmark,Maison} and follows closely the 4-dimensional case. However, there are several new difficulties associated with the analysis of the resulting problem. First, even without angular momenta the problem is nonlinear, in contrast to the linear structure present in the static 4D setting. This makes the construction of a \emph{model map} prescribing the singular behavior near $\Gamma$
much more delicate, whereas in the 4D case the superposition of Schwarzschild solutions is sufficient. Next, the target $SL(3,\R)/SO(3)$ is a rank 2 symmetric space with nonpositive sectional curvature, rather than rank 1 with negative sectional curvature in 4D.  We recall that the theory of harmonic maps into rank 1 symmetric spaces, in particular real hyperbolic space $\mathbb{H}^n$, has been extensively investigated e.g. \cite{Schoen, LiTam}, yet comparatively little is known for the cases of higher rank targets.  These properties of the target hyperbolic space $\mathbb{H}^2= SL(2,\R)/SO(2)$ in dimension four played a central role in obtaining a priori estimates to prove existence, and without these properties in the 5D case new techniques must be developed. Furthermore, in higher dimensions there is an abundance of possible rod structures, and they must obey an admissibility condition \eqref{admissibility0} not present in four dimensions. Finally, the study of conical singularities and their formulation as the balancing condition $b_i=0$, while similar to the 4D case, requires a more precise analysis.

Several explicit solutions of these equations and related ones have previously been found. As mentioned above, the Myers-Perry black hole \cite{MyersPerry} generalizes the Kerr black hole to 5-dimensions, and is a 3-parameter family of solutions with spherical $S^{3}$ horizon topology. Emparan and Reall \cite{EmparanReall} found the first example with nontrivial topology, namely a family of black ring solutions with an $S^1\times S^2$ horizon and one angular momentum. These were later generalized by Pomeransky-Sen'kov \cite{PomeranskySenkov} to a full 3-parameter family with two angular momenta. A multiple horizon solution with two components consisting of an $S^3$ surrounded by an $S^1\times S^2$, referred to as black saturn, was constructed by Elvang and Figueras \cite{ElvangFigueras}. In this solution both the sphere and ring rotate only in one plane which is associated with the $S^1$ direction of the ring.
Further multiple horizon solutions include the dipole black rings (or di-rings) \cites{EvslinKrishnan,IguchiMishima} consisting of two concentric singly spinning rings rotating in the same plane, and the bicycling black rings (or bi-rings) \cites{ElvangRodriguez,Izumi} consisting of two singly spinning rings rotating in orthogonal planes. In the minimal supergravity setting, Kunduri and Lucietti \cite{LuciettiKunduri} found the first examples of regular black holes having a lens space topology $\mathbb{RP}^3=L(2,1)$. These were generalized by Tomizawa and Nozawa to more general lens topology $L(p,1)$ in \cite{TomizawaNozawa}. Both of these black lens solutions are supersymmetric and hence extremal. It is an important open problem to
find regular vacuum black holes with lens topology. In this direction Chen and Teo \cite{ChenTeo} found vacuum black lenses via the inverse scattering method, however their solutions either possess conical singularities or have a naked singularity.
A disadvantage of the methods used to construct the above examples
is that they cannot produce all possible regular solutions. In contrast, the PDE approach used here generates all candidates with an admissible/compatible rod structure, where the only obstruction is the possibility of conical singularities on the bounded components of the axes. Furthermore, the variety of black holes that may be constructed from admissible rod data which also satisfy the compatibility condition is vast. In particular, multiple and single component black lenses $L(p,q)$ are possible, for arbitrary relatively prime $p$ and $q$, as is shown in Proposition \ref{lensrod} of Section \ref{sec4}.

The existence portion of Theorem \ref{main} may be generalized by forgoing the admissibility condition \eqref{admissibility0}. This requires instead of \eqref{compatibilitycondition} a \textit{generalized compatibility condition}
\begin{equation}\label{gcompatibilitycondition0}
m_{l-1}m_{l+1} \det\begin{pmatrix} m_{l-1} & n_{l-1} \\ m_{l} & n_{l} \end{pmatrix}
\det\begin{pmatrix} m_l & n_l \\ m_{l+1} & n_{l+1} \end{pmatrix} \leq 0,
\end{equation}
which is used in the construction of a model map. Note that if \eqref{admissibility0} is satisfied then \eqref{gcompatibilitycondition0} reduces to \eqref{compatibilitycondition}. However, without the admissibility condition orbifold singularities at corner points will be present.

\begin{theorem}\label{main2}
Given a rod data set $\D$ satisfying the generalized compatibility condition \eqref{gcompatibilitycondition0}, there is a unique harmonic map $\varphi\colon\R^3\setminus\Gamma\to SL(3,\R)/SO(3)$ with prescribed singularities on $\Gamma$ corresponding to $\D$. From this map a well-behaved 5-dimensional asymptotically flat, stationary, bi-axially symmetric solution of the vacuum Einstein equations without degenerate horizons can be constructed having orbifold singularities at the corners if and only if the resulting metric coefficients are sufficiently smooth across $\Gamma$ and $b_l=0$ on any bounded axis rod.
\end{theorem}

This result has been generalized in \cite{KhuriWeinsteinYamada}  to include the asymptotically Kaluza-Klein and asymptotically locally Euclidean cases, in which cross sections at infinity are $S^1 \times S^2$ and quotients of $S^3$ respectively.

The organization of this paper is as follows.
In Section \ref{sec2} we review the reduction of the Einstein vacuum equations, in the bi-axially symmetric stationary setting, to a harmonic map having the symmetric space $SL(3,\R)/SO(3)$ as target. Relevant aspects of the geometry of this symmetric space are then discussed in Section \ref{sec3}.
In Section \ref{sec4} a detailed analysis of rod structures and the hypotheses associated with them is given. The model map is constructed in Section \ref{sec5}, and existence and uniqueness for the harmonic map problem is proven in Section \ref{sec7} using energy estimates established in Section \ref{sec6}. Finally in Section \ref{sec8} it is shown how the desired spacetime is produced from the harmonic map, and regularity issues are discussed. An appendix is included in order to give a topological characterization of corners.

\medskip

\textbf{Acknowledgements.}
The authors thank the Erwin Schr\"odinger International Institute for Mathematics and Physics and the organizers of its ``Geometry and Relativity'' program, where portions of this paper were written. The third author also thanks Koichi Kaizuka for useful conversations concerning the geometry of symmetric spaces.

\section{Dimensional Reduction of the Vacuum Einstein Equations} \label{setup}
\label{sec2} \setcounter{equation}{0}
\setcounter{section}{2}

Let $\mathcal{M}^5$ be the domain of outer communication for a well-behaved asymptotically flat, stationary vacuum, bi-axisymmetric spacetime. In particular its isometry group admits
$\mathbb{R}\times U(1)^2$ as a subgroup in which the $\mathbb{R}$-generator $\xi$ (time translation) is timelike in the asymptotic region, and the $U(1)^2$-generators $\eta^{(i)}$, $i=1,2$ yield spatial rotation.
Since the three generators for the isometry subgroup commute, they may be expressed as coordinate vector fields $\xi=\partial_{t}$ and $\eta^{(i)}=\partial_{\phi^{i}}$.
Moreover by abusing notation so that the same symbols denote dual covectors it holds that
\begin{equation}\label{aaa}
\star\left(\xi\wedge\eta^{(1)}\wedge\eta^{(2)}\wedge d\xi\right)
=\star\left(\xi\wedge\eta^{(1)}\wedge\eta^{(2)}\wedge d\eta^{(1)}\right)
=\star\left(\xi\wedge\eta^{(1)}\wedge\eta^{(2)}\wedge d\eta^{(2)}\right)
=0,
\end{equation}
where $\star$ denotes the Hodge star operation. To see this, observe that the vacuum
equations imply that the exterior derivative of the three quantities in \eqref{aaa} vanishes, and since these functions vanish on the axis in the asymptotically flat end they must vanish everywhere. Therefore the Frobenius theorem applies to show that the 2-plane distribution orthogonal to the three Killing vectors is integrable. We may then take coordinates on one of these resulting 2-dimensional orbit manifolds, and Lie drag them to get a system of coordinates such that the spacetime metric decomposes in the following way
\begin{equation}
g=\sum_{a,b=1}^{3}q_{ab}(x)dy^{a}dy^{b}+\sum_{c,d=4}^{5}h_{cd}(x)dx^{c}dx^{d},
\end{equation}
where $y=(\phi^1,\phi^2,t)$. The fiber metric may be expressed by
\begin{equation}\label{fibermetric}
q=f_{ij}(d\phi^{i}+v^{i}dt)(d\phi^{j}+v^{j}dt)-f^{-1}\rho^2 dt^{2},
\end{equation}
for some functions $v^i$ where $f=\det f_{ij}$ and $\rho^2=-\det q_{ab}$. It is shown in \cites{Chrusciel,ChruscielCosta} that the determinant of the fiber metric is nonpositive, and
the vacuum equations imply that $\rho$ is harmonic with respect to the metric $fh$, since
\begin{equation}
\Delta_{fh}\rho=\rho^{-1}R_{tt}-\rho f^{-1}f^{ij} R_{ij}=0.
\end{equation}
From this it can be shown \cites{Chrusciel,ChruscielCosta} that $\rho$ is a well-defined coordinate function on the quotient
$\mathcal{M}^{5}/\left[\mathbb{R}\times U(1)^2\right]$ away from the poles, that is $|\nabla\rho|\neq 0$. Since the orbit space is simply connected \cite{HollandsYazadjiev1} there is a globally defined harmonic conjugate function $z$, which together with $\rho$ yields an isothermal coordinate system so that
\begin{equation}
fh=e^{2\sigma}(d\rho^2 +dz^2),
\end{equation}
for some function $\sigma=\sigma(\rho,z)$. We now have the canonical Weyl-Papapetrou expression for the spacetime metric
\begin{equation} \label{metric}
g=f^{-1}e^{2\sigma}(d\rho^2+dz^2)-f^{-1}\rho^2 dt^2
+f_{ij}(d\phi^{i}+v^{i}dt)(d\phi^{j}+v^{j}dt).
\end{equation}

Let
\begin{equation}
g_{3}=e^{2\sigma}(d\rho^2+dz^2)-\rho^2 dt^2,\quad\quad\quad
A^{(i)}=v^{i}dt,
\end{equation}
then
\begin{equation}
g=f^{-1}g_{3}+f_{ij}(d\phi^{i}+A^{(i)})(d\phi^{j}+A^{(j)}).
\end{equation}
This represents a Kaluza-Klein reduction with 2-torus fibers. In this setting the vacuum Einstein equations yield a 3-dimensional version of Einstein-Maxwell theory, with the `Maxwell equations' given by
\begin{equation}\label{maxwell}
d(ff_{ij}\star_{3} dA^{(j)})=0,
\end{equation}
where $\star_{3}$ is the Hodge star operation with respect to $g_{3}$. It follows that there exist globally defined (due to simple connectivity) twist potentials satisfying
\begin{equation}\label{chi}
d\omega_{i}=2f f_{ij}\star_{3}dA^{(j)}.
\end{equation}
In particular if $v^i$ are constant then the potentials $\omega_i$ are constant, and vice versa. To explain the geometric meaning of the forms appearing on the right-hand side of \eqref{chi}
observe that $\eta^{(i)}=f_{ij}\left(d\phi^j+v^j dt\right)$ is the dual 1-form to $\partial_{\phi^{i}}$, and according to Frobenius' theorem
the forms $\eta^{(1)}\wedge\eta^{(2)}\wedge d\eta^{(i)}$
measure the lack of integrability of the orthogonal complement distribution to the axisymmetric Killing fields. Moreover, it turns out that these forms are directly related to \eqref{chi}. Indeed let $\epsilon$, $\epsilon_{3}$, and $\star_{3}$ denote the volume forms with respect to $g$ and $g_{3}$, and the Hodge star operator with respect to $g_{3}$, respectively, then since
\begin{equation}
d\eta^{(i)}=f_{ij}d A^{(j)}+df_{ij}\wedge\left(f^{ja}\eta^{(a)}\right)
\end{equation}
we have
\begin{align}\label{komar}
\begin{split}
\star(\eta^{(1)}\wedge\eta^{(2)}\wedge d\eta^{(i)})=&
f_{ij}\star(\eta^{(1)}\wedge\eta^{(2)}\wedge dA^{(j)})\\
=&f_{ij}\epsilon(\text{ }\!\cdot\text{ }\!,\partial_{\phi^1},\partial_{\phi^2},\partial_{l},\partial_{k})
\left(dA^{(j)}\right)^{lk}\\
=&f^{-1}f_{ij}\epsilon_{3}(\text{ }\!\cdot\text{ }\!,\partial_{l},\partial_{k})\left(dA^{(j)}\right)^{lk}\\
=&ff_{ij}\star_{3} dA^{(j)}.
\end{split}
\end{align}
Note also that since the spacetime is vacuum and $\eta^{(i)}$ are dual to Killing fields,
standard computations along with Cartan's `magic' formula show that the 1-forms
\begin{equation}\label{komar1}
\star(\eta^{(1)}\wedge\eta^{(2)}\wedge d\eta^{(i)})=\iota_{\eta^{(1)}}
\iota_{\eta^{(2)}}\star d\eta^{(i)}
\end{equation}
are closed, where $\iota$ denotes interior product. This yields an alternate proof of \eqref{maxwell}, and confirms that the twist potentials $\omega_i$ agree with those associated with the Komar expression for angular momentum.

Next, following Maison \cite{Maison} define the following $3\times 3$ matrix
\begin{equation}\label{bigmatrix}
\Phi=
\left(
    \begin{array}{ccc}
      f^{-1} & - f^{-1} \omega_1 & - f^{-1} \omega_2 \\
      -f^{-1} \omega_1 & f_{11} + f^{-1} \omega_1^2  & f_{12} + f^{-1} \omega_1 \omega_2  \\
      -f^{-1} \omega_2 & f_{12} + f^{-1} \omega_1 \omega_2  & f_{22} + f^{-1} \omega_2^2
    \end{array}
  \right)
\end{equation}
which is symmetric, positive definite, and has $\det\Phi=1$.
The inverse matrix is
\begin{equation}
\Phi^{-1}=
\left(
    \begin{array}{ccc}
      f + f^{11} \omega_1^2 + f^{22} \omega_2^2 + 2 f^{12} \omega_1 \omega_2 & f^{11} \omega_1 + f^{21} \omega_2  & f^{12} \omega_1 + f^{22} \omega_2 \\
      f^{11} \omega_1 + f^{12} \omega_2 & f^{11}  & f^{12}   \\
      f^{21} \omega_1 + f^{22} \omega_2   &f^{21} & f^{22}    \end{array}
  \right)
 .
\end{equation}
This allows for a simplified expression of the 3-dimensional reduced Einstein-Hilbert action
\begin{equation}\label{action}
\mathcal{S}=\int_{\mathbb{R}\times \left(\mathcal{M}^{5}/[\mathbb{R}\times U(1)^2]\right)}
R^{(3)}\star_{3}1+\frac{1}{4}\mathrm{Tr}\left(\Phi^{-1}d\Phi\wedge\star_{3}\Phi^{-1}d\Phi\right).
\end{equation}
The Einstein-harmonic map system arising from this action is
\begin{equation}\label{einstein}
R^{(3)}_{kl}-\frac{1}{2}R^{(3)}(g_3)_{kl}=T_{kl},\quad\quad \operatorname{div}_{\mathbb{R}^3}\left(\Phi^{-1}
\nabla\Phi\right)=0,
\end{equation}
where the stress-energy tensor for the harmonic map is
\begin{equation}
T_{kl}=\mathrm{Tr}\left(J_{k}J_{l}\right)
-\frac{1}{2}g_{3}^{mn}\mathrm{Tr}\left(J_{m}J_{n}\right)
(g_{3})_{kl}
\end{equation}
with the current
\begin{equation}
J_{l}=\Phi^{-1}\partial_{l}\Phi.
\end{equation}
Note that by taking a trace the Einstein equations may be reexpressed as
\begin{equation}
R^{(3)}_{kl}=\mathrm{Tr}\left(J_{k}J_{l}\right).
\end{equation}
Furthermore, in the $\Phi$ portion of the action cancelations occur so that $e^{2\sigma}$ does not appear, and this results in the divergence of \eqref{einstein} with respect to
the Euclidean metric
\begin{equation}\label{flatmetric}
\delta=d\rho^2+dz^2+\rho^2 d\phi^2,
\end{equation}
where $\phi$ is an auxiliary coordinate.
This also implies that the stress-energy tensor is divergence free with respect to the Euclidean metric
\begin{equation}\label{divergence}
0=\left(\operatorname{div}_{\mathbb{R}^3}T\right)(\partial_{\rho})
=\partial_{\rho}(\rho T_{\rho\rho})+\partial_{z}(\rho T_{\rho z}),\quad\quad
0=\left(\operatorname{div}_{\mathbb{R}^3}T\right)(\partial_{z})
=\partial_{\rho}(\rho T_{\rho z})+\partial_{z}(\rho T_{z z}).
\end{equation}

The divergence free property of $T$ follows from the harmonic map equations. To see this in a more general harmonic map setting, consider maps $\varphi: (M,\mathrm{g})\rightarrow (N,\mathrm{h})$ with harmonic energy
\begin{equation}
E=\frac{1}{2}\int_{M}|d\varphi|^2 dx_{\mathrm{g}}=\frac{1}{2}\int_{M}\mathrm{g}^{\mathrm{ij}}
\mathrm{h}_{\mathrm{lk}}\partial_{\mathrm{i}}\varphi^{\mathrm{l}}
\partial_{\mathrm{j}}\varphi^{\mathrm{k}} dx_{\mathrm{g}}.
\end{equation}
The first variation is given by
\begin{equation}
\frac{\delta E}{\delta \mathrm{g}}=\frac{1}{2}\int_{M}\delta \mathrm{g}^{\mathrm{ij}}\left(\mathrm{h}_{\mathrm{lk}}
\partial_{\mathrm{i}}\varphi^{\mathrm{l}}\partial_{\mathrm{j}}\varphi^{\mathrm{k}}
-\frac{1}{2}|d\varphi|^2 \mathrm{g}_{\mathrm{ij}}\right) dx_{\mathrm{g}},
\end{equation}
and the stress-energy tensor is
\begin{equation}
T_{\mathrm{ij}}=\langle\partial_{\mathrm{i}}\varphi,\partial_{\mathrm{j}}
\varphi\rangle_{\mathrm{h}}
-\frac{1}{2}|d\varphi|^2 \mathrm{g}_{\mathrm{ij}}.
\end{equation}
The harmonic map equations
\begin{equation}\label{tensiondef}
\tau(\varphi)=\hat{\nabla}^{\mathrm{i}}\partial_{\mathrm{i}}\varphi=0
\end{equation}
then imply that the stress-energy tensor is divergence free
\begin{align}
\begin{split}
\nabla^{\mathrm{i}}T_{\mathrm{ij}}=\langle\hat{\nabla}^{\mathrm{i}}
\partial_{\mathrm{i}}\varphi,\partial_{\mathrm{j}}\varphi\rangle_{\mathrm{h}}
+\langle\partial_{\mathrm{i}}\varphi,\hat{\nabla}^{\mathrm{i}}
\partial_{\mathrm{j}}\varphi\rangle_{\mathrm{h}}
-\mathrm{g}^{\mathrm{lm}}\langle\hat{\nabla}_{\mathrm{j}}
\partial_{\mathrm{l}}\varphi,\partial_{\mathrm{m}}\varphi\rangle_{\mathrm{h}}=0.
\end{split}
\end{align}
Here $\hat{\nabla}$ is the induced connection on the bundle $T^{*}M\otimes \varphi^{-1} TN$, and $\tau(\varphi)$ denotes the tension field which is a section of the pullback bundle $\varphi^{-1}TN$.

The Einstein equations of \eqref{einstein} may be solved via quadrature. This may be shown by computing each equation in terms of metric components. Recall that
\begin{equation}
R_{kl}^{(3)}=\partial_{m}\Gamma_{kl}^{m}-\partial_{l}\Gamma_{km}^{m}
+\Gamma_{kl}^{m}\Gamma_{nm}^{n}-\Gamma_{kn}^{m}\Gamma_{lm}^{n},
\end{equation}
and
\begin{equation}
R^{(3)}=g_3^{kl}R_{kl}^{(3)}
=-\rho^{-2}R_{tt}^{(3)}+e^{-2\sigma}\left(R_{\rho\rho}^{(3)}
+R_{zz}^{(3)}\right).
\end{equation}
The Christoffel symbols are
\begin{equation}
\Gamma_{tt}^{l}=\delta^{l\rho}e^{-2\sigma}\rho,\quad\quad
\Gamma_{ti}^{l}=\delta_{t}^{l}\delta_{i}^{\rho}\rho^{-1},\quad\quad
\Gamma_{ij}^{l}=\delta_{j}^{l}\partial_{i}\sigma+\delta_{i}^{l}\partial_{j}\sigma
-\delta_{ij}\delta^{lm}\partial_{m}\sigma\text{ }\text{ }\text{ for }\text{ }\text{ }i,j\neq t.
\end{equation}
It follows that
\begin{equation}
R_{tt}^{(3)}=R_{ti}^{(3)}=0, \quad i\neq t,\quad
R_{\rho\rho}^{(3)}=-\Delta_{\mathbb{R}^2}\sigma+\frac{1}{\rho}\partial_{\rho}\sigma,\quad
R_{zz}^{(3)}=-\Delta_{\mathbb{R}^2}\sigma-\frac{1}{\rho}\partial_{\rho}\sigma,\quad
R_{\rho z}^{(3)}=\frac{1}{\rho}\partial_{z}\sigma.
\end{equation}
From this the quadrature equations for $\sigma$ are found to be
\begin{equation}
\partial_{\rho}\sigma=\frac{\rho}{2}\left(R_{\rho\rho}^{(3)}-R_{zz}^{(3)}\right)
=\frac{\rho}{2}\left(\mathrm{Tr}(J_{\rho}J_{\rho})-\mathrm{Tr}(J_{z}J_{z})\right)
=\rho T_{\rho\rho}=-\rho T_{zz},
\end{equation}
\begin{equation}
\partial_{z}\sigma=\rho R_{\rho z}^{(3)}=\rho T_{\rho z},
\end{equation}
which may be rewritten more conveniently as
\begin{equation}\label{sigma}
d\sigma=-\iota_{\eta}\ast\iota_{\partial_{z}}T
\end{equation}
where $\ast$ is the Hodge star operation with respect to the metric $\delta$ on $\mathbb{R}^3$, and $\eta=\partial_{\phi}$. To see this let $\varepsilon$ denote the volume form for $\delta$, then
\begin{equation}
(\ast\iota_{\partial_{z}}T)_{ij}=\varepsilon_{ijl}T^{lz}
\end{equation}
and hence
\begin{equation}
(\iota_{\eta}\ast\iota_{\partial_{z}}T)_{j}=\varepsilon_{ijl}\eta^{i}T^{lz}
=\varepsilon(\partial_{\phi},\partial_{j},\partial_{\rho})T^{\rho z}
+\varepsilon(\partial_{\phi},\partial_{j},\partial_{z})T^{zz}.
\end{equation}
We then have
\begin{equation}\label{333}
\iota_{\eta}\ast\iota_{\partial_{z}}T
=\rho T_{zz} d\rho -\rho T_{\rho z} dz,
\end{equation}
which confirms \eqref{sigma}. Moreover, for later use observe that this form is closed in light of the harmonic map equations
\begin{equation}
d\left(\iota_{\eta}\ast\iota_{\partial_{z}}T\right)
=-\left(\operatorname{div}_{\mathbb{R}^3}T\right)(\partial_{z})d\rho\wedge dz=0.
\end{equation}
Note that we also have to show that $\sigma$ obtained from quadrature is bi-axisymmetric. However this follows easily from \eqref{333}, since
\begin{equation}
\iota_{\eta^{(i)}}d\sigma=\iota_{\eta^{(i)}}\iota_{\eta}\ast\iota_{\partial_{z}}T=0.
\end{equation}

\section{The Riemannian Geometry of  $SL(3,\R)/SO(3)$}
\label{sec3} \setcounter{equation}{0}
\setcounter{section}{3}




The harmonic map arising from the dimensional reduction of the bi-axisymmetric stationary vacuum Einstein equations has as target space $SL(3,\mathbb{R})/SO(3)$. The geometry of this symmetric space plays an important role in the analysis of the harmonic map, and in this section the relevant aspects will be described.

Let $G=SL(3,\mathbb{R})$ then $K=SO(3)$ is a maximal compact subgroup. The quotient $\mathbf{X}=G/K$ is the space of equivalence classes $[A]$ in which
\begin{equation}
A \in SL(3, {\R}) \mbox{ and  } A \sim A' \Leftrightarrow A' = A B  \mbox{ for some } B \in SO(3).
\end{equation}
In other words $\mathbf{X}$ is the space of left cosets of $K$ in $G$ and $G$ acts transitively on $\mathbf{X}$ by
\begin{equation}
A'K\mapsto AA' K\quad\text{ }\text{ for }\text{ }\quad A\in G,
\end{equation}
so that $K$ is the isotropy subgroup at $x_0=[\mathrm{Id}]$.
Recall now the construction of the canonical $G$-invariant Riemannian metric on the homogeneous space $\mathbf{X}$, which yields a Riemannian symmetric space structure.
The Lie algebras will be denoted by
\begin{equation}
\mathfrak{g} = sl(3) = \{Y \in gl (3) \,\, | \,\, \mathrm{Tr} Y = 0\},
\end{equation}
and
\begin{equation}
\mathfrak{k} = so(3) = \{Y \in gl(3) \,\, | \,\, Y^t = - Y\}.
\end{equation}
Note that $\mathfrak{g}$ is semisimple since the Killing form $\mathbf{B}:\mathfrak{g}\times\mathfrak{g}\rightarrow\mathbb{R}$ given by
\begin{equation}
\mathbf{B}(Y, Z) = \mathrm{Tr} (\mathrm{ad} Y \circ \mathrm{ad} Z) = 6\mathrm{Tr}(YZ)
\end{equation}
is nondegenerate.  Let $\mathfrak{p}$ be the orthogonal complement of $\mathfrak{k}$ with respect to $\mathbf{B}$, so that we have the Cartan decomposition
\begin{equation}
\mathfrak{g} = \mathfrak{k} \oplus \mathfrak{p}
\end{equation}
with
\begin{equation}
\mathfrak{p} = \{Y \in gl (3) \,\, | \,\, Y^t = Y, \text{ }\mathrm{Tr} Y = 0\},
\end{equation}
and satisfying the Cartan relations
\begin{equation}
[\mathfrak{k}, \mathfrak{k}] \subset \mathfrak{k}, \quad\quad
[\mathfrak{p}, \mathfrak{p} ] \subset \mathfrak{k}, \quad\quad
[\mathfrak{k},\mathfrak{p}]\subset\mathfrak{p}.
\end{equation}
The Killing form $\mathbf{B}$ is negative definite on $\mathfrak{k}$ and positive definite on $\mathfrak{p}$, in particular $\mathbf{X}$ is of noncompact type.

Consider the Cartan involution $\theta: \mathfrak{g} \rightarrow \mathfrak{g}$ with $\theta|_{\mathfrak{k}} = \mathrm{id}, \theta|_{\mathfrak{p}} = - \mathrm{id}$, where in our context $\theta(Y) = - Y^t$. Then the quadratic form
\begin{equation}
\langle Y, Z\rangle_{\mathfrak{g}}=
\begin{cases}
-\frac{2}{3}\mathbf{B}(Y,Z) & \text{ if } Y,Z\in\mathfrak{k},\\
- \frac{2}{3}\mathbf{B}(Y, \theta(Z))& \text{ if } Y,Z\in\mathfrak{p},\\
0& \text{ if } Y\in\mathfrak{k},\text{ } Z\in\mathfrak{p},
\end{cases}
\end{equation}
is positive definite and Ad $K$-invariant. From this the desired Riemannian metric at $x_0$ is obtained by restricting the quadratic form to $\mathfrak{p}$ which is identified with $T_{x_0} \mathbf{X}$, namely
\begin{equation}
 \mathbf{g}_{x_0}(Y, Z) = 4 \mathrm{Tr} (YZ^t) \quad \text{ for }\quad Y, Z \in \mathfrak{p}.
\end{equation}
This in turn gives rise to the metric globally on $\mathbf{X}$ via left translation. Let $L_B: \mathbf{X} \rightarrow \mathbf{X}$ denote the left translation operator
\begin{equation}
L_B(x) = L_B([A]) = [B A],
\end{equation}
where $A,B\in SL(3,\R)$ and $x=[A]$.
Since $SL(3, {\mathbb R})$ acts transitively on $\mathbf{X}$, given $x \in \mathbf{X}$ there is a $B\in SL(3, {\mathbb R})$ such
that $L_B(x_0) = x$, and thus the $G$-invariant Riemannian metric at $x$ may be defined
by pulling back the quadratic form at the identity
\begin{equation}\label{symmetricspacemetric}
\mathbf{g}_x = L_{B^{-1}}^* \mathbf{g}_{x_0}.
\end{equation}
With this metric $SL(3,\R)/SO(3)$ becomes a symmetric space of noncompact type having rank 2 (see \cite{BallmanGromovSchroeder}). In particular it has nonpositive curvature, with the sectional curvature of the plane spanned by orthonormal vectors $ Y, Z \in \mathfrak{p}$ given by $-\parallel[Y,Z]\parallel_{\mathfrak{g}}^2$.

In order to connect the metric \eqref{symmetricspacemetric} with the target space geometry associated to the harmonic map of the previous section, the following characterization of $\mathbf{X}=SL(3,\R)/SO(3)$ will be needed. Recall the polar decomposition for matrices, namely any $A\in SL(3,\R)$ may be written uniquely as $A=PO$ where $O\in SO(3)$ and $P\in \tilde{\mathbf{X}}$ with
\begin{equation}
\tilde{\mathbf{X}} = \{A \in SL(3, {\mathbb R}) \mid A \text{ is symmetric and positive definite}\}.
\end{equation}
This indicates that $\mathbf{X}$ may be identified with $\tilde{\mathbf{X}}$, and in fact this is accomplished with the map $\mathcal{I} : \tilde{\mathbf{X}} \rightarrow \mathbf{X}$ given by
\begin{equation}
\mathcal{I}(A) = [A^{1/2}],\quad\quad\quad \mathcal{I}^{-1}([B])=BB^t.
\end{equation}
Observe that $\tilde{\mathbf{X}}$ can be interpreted as the set of all ellipsoids in ${\mathbb R}^3$ centered at the origin with unit volume, and is diffeomorphic to $\mathbb{R}^5$ (hence the same is true of $\mathbf{X}$). Moreover
$SL(3,{\R})$ acts transitively on $\tilde{\mathbf{X}}$ by the analogue of left translation $\tilde{L}_B = \mathcal{I}^{-1} \circ L_B \circ \mathcal{I}$, that is
\begin{equation}
\tilde{L}_B (A) = B A B^t.
\end{equation}

The identification above naturally induces a pull-back metric $\tilde{\mathbf{g}}:= \mathcal{I}^* \mathbf{g}$ on $\tilde{\mathbf{X}}$. At the identity this is
\begin{equation}\label{identitymetric}
\tilde{\mathbf{g}}_{\mathrm{Id}} (V, V) = \mathbf{g}_{x_0} \left(\frac{V}{2}, \frac{V}{2}\right) =   \mathrm{Tr}(VV^t),
\end{equation}
for
\begin{equation}
V\in T_{\mathrm{Id}} \tilde{\mathbf{X}} = \{W \in Mat_{3 \times 3} (\mathbb{R}) \,\, | \,\, W^t = W, \quad \mathrm{Tr} W = 0 \}.
\end{equation}
As for an arbitrary point $A\in\tilde{\mathbf{X}}$ and $V\in T_{A}\tilde{\mathbf{X}}$,
\begin{align}
\begin{split}
\tilde{\mathbf{g}}_{A}(V,V)=&\mathbf{g}_{\mathcal{I}(A)}\left(d\mathcal{I}_{A}(V),
d\mathcal{I}_{A}(V)\right)\\
=&L_{A^{-1/2}}^{*}\mathbf{g}_{x_0}\left(d\mathcal{I}_{A}(V),d\mathcal{I}_{A}(V)\right)\\
=&\mathbf{g}_{x_0}\left(d(L_{A^{-1/2}}\circ\mathcal{I})_{A}(V),
d(L_{A^{-1/2}}\circ\mathcal{I})_{A}(V)\right)\\
=&\mathrm{Tr}\left([(d\tilde{L}_{A^{-1/2}})_{A}(V)][(d\tilde{L}_{A^{-1/2}})_{A}(V)]^{t}\right).
\end{split}
\end{align}
Since
\begin{equation}
(d\tilde{L}_{A^{-1/2}})_{A}(V)=A^{-1/2}V(A^{-1/2})^t,
\end{equation}
it follows that
\begin{align}\label{targetmetric}
\begin{split}
\tilde{\mathbf{g}}_{A}(V,V)=&\mathrm{Tr}\left(A^{-1/2}V(A^{-1/2})^t A^{-1/2} V(A^{-1/2})^t\right)\\
=&\mathrm{Tr}\left(A^{-1/2}VA^{-1}V(A^{-1/2})^t\right)\\
=&\mathrm{Tr}\left(A^{-1}V A^{-1}V\right).
\end{split}
\end{align}

Recall from the previous section that a given 5-dimensional bi-axisymmetric stationary vacuum spacetime yields a map $\Phi:\mathbb{R}^3\setminus\Gamma\rightarrow\tilde{\mathbf{X}}$, where $\mathbb{R}^3$ is parameterized by the Weyl-Papapetrou coordinates $(\rho,z,\phi)$, $\Gamma$ denotes the $z$-axis, and $\tilde{\mathbf{X}}$ is parameterized by $(f_{ij},\omega_{i})$. According to \eqref{targetmetric} the pull-back metric is then given by
\begin{equation}
\Phi^* \tilde{\mathbf{g}} =  \mathrm{Tr} ( \Phi^{-1}d \Phi \, \Phi^{-1} d \Phi).
\end{equation}
Since this agrees with the expression appearing in the reduced action \eqref{action}, it follows that the bi-axisymmetric stationary vacuum Einstein equations reduce to a harmonic map problem with target space $SL(3,\R)/SO(3)$.

\section{The Rod Structure}
\label{sec4} \setcounter{equation}{0}
\setcounter{section}{4}

A well-behaved asymptotically flat, stationary vacuum, bi-axisymmetric spacetime admits a global system of Weyl-Papapetrou coordinates in its domain of outer communication $\mathcal{M}^5$, as described in Section \ref{sec2}, in which the metric takes the form
\begin{equation}\label{spacetimemetric}
g=f^{-1}e^{2\sigma}(d\rho^2+dz^2)-f^{-1}\rho^2 dt^2
+f_{ij}(d\phi^{i}+v^{i}dt)(d\phi^{j}+v^{j}dt).
\end{equation}
The orbit space $\mathcal{M}^{5}/[\mathbb{R}\times U(1)^2]$ is diffeomorphic to the right-half plane $\{(\rho,z)\mid \rho>0\}$ (see \cite{HollandsYazadjiev1}), and
its boundary $\rho=0$ encodes nontrivial aspects of the topology. Let $q$ be the fiber metric \eqref{fibermetric} consisting of the last two terms in \eqref{spacetimemetric}. In order to avoid curvature singularities $\mathrm{dim} \left(\mathrm{ker}\text{ }\! q(0,z)\right)=1$ except at isolated points $p_{l}$, $l=1,\ldots,L$ where the dimension of the kernel is 2 \cites{Harmark,HollandsYazadjiev}. It follows that the $z$-axis is broken into $L+1$ intervals called rods
\begin{equation}
\Gamma_{1}=[z_{1},\infty),\text{ }\Gamma_{2}=[z_2,z_1],\text{ }\ldots,\text{ }
\Gamma_{L}=[z_{L},z_{L-1}],\text{ }\Gamma_{L+1}=(-\infty,z_{L}],
\end{equation}
on which either $|\partial_{t}+\Omega_{1}\partial_{\phi^1}+\Omega_{2}\partial_{\phi^2 }|$ vanishes (horizon rod) or $(f_{ij})$ fails to be of full rank (axis rod). Here $\Omega_i$ denotes the angular velocity of the horizon and is given by $-v^i$ restricted to the rod. This must be a constant, and can be seen by solving for $dv^i$ from \eqref{chi} and showing that it vanishes on the rod. The condition for an axis rod implies \cite{HollandsYazadjiev} that for each such $\Gamma_{l}$ there is a pair of relatively prime integers $(m_{l},n_{l})$ so that the Killing field
\begin{equation}
m_{l}\partial_{\phi^1}+n_{l}\partial_{\phi^2}
\end{equation}
vanishes on $\Gamma_{l}$. Observe that $m_{l}$ and $n_{l}$ must be integers since elements of the isotropy subgroup at the axis are of the form $(e^{im_{l}\phi},e^{in_{l}\phi})$, $0\leq\phi<2\pi$, and the isotropy subgroup forms a proper closed subgroup of $T^2=S^1\times S^1$. That is, the isotropy subgroup yields a simple closed curve in the torus exactly when the slope of its winding is rational.
The pair $(m_{l},n_{l})$ is referred to as the rod structure for the rod $\Gamma_{l}$, and $(0,0)$ serves as the rod structure for any horizon rod. Note that the rod structure is not unique in terms of the information that it encodes, although this type of uniqueness is valid when the rod structure is viewed as an element of $\mathbb{RP}^1$.

The asymptotically flat condition is encoded by the rod structures of $\Gamma_{1}$ and $\Gamma_{L+1}$ by requiring them to be $(\pm 1,0)$ and $(0,\pm 1)$ or vice versa. This, of course, arises from the rod structure of Minkowski space $\mathbb{R}^{4,1}$ which will now be described in order to motivate the definition of a `corner'. The Weyl-Papapetrou form of the Minkowski metric is derived from the polar coordinate expression with the help of Hopf coordinates $(\theta,\phi^1,\phi^2)$, $\phi^{i}\in[0,2\pi]$, $\theta\in[0,\pi/2]$ on the 3-sphere and a conformal mapping
\begin{align}\label{Minkowski}
\begin{split}
  g_{0} = & -dt^{2}+ dr^2 + r^2 d \omega_{S^3}^2 \\
    = & -dt^{2}+ dr^2 + r^2 \left[d \theta^2 + \sin^2 \theta (d \phi^1)^2 + \cos^2 \theta (d \phi^ 2)^2 \right] \\
    = & q_{0} +  dr^2 + r^2 d \theta^2 \\
    = & q_{0} + \frac{1}{4\sqrt{\rho^2 + z^2 }} (d \rho^2  + dz^2).
\end{split}
\end{align}
Here the conformal map in the complex plane is given by
\begin{equation}
\zeta \mapsto  \zeta^2 \quad :\quad {\mathbb R}_{\geq 0} \times  {\mathbb R}_{\geq 0} \rightarrow \mbox{$\rho z$-half plane},
\end{equation}
or rather
\begin{equation}
\rho =  r^2 \sin 2 \theta, \quad\quad z =  r^2 \cos 2 \theta .
\end{equation}
If $x^i$ denote cartesian coordinates then the Killing fields
\begin{equation}
\partial_{\phi^1} = -x^2 \partial_{x^1} + x^1 \partial_{x^2},\quad\quad\quad
\partial_{\phi^2} = -x^4 \partial_{x^3} + x^3 \partial_{x^4},
\end{equation}
vanish on the rods $\Gamma_{1}=[0,\infty)$ and $\Gamma_{2}=(-\infty,0]$, respectively. Therefore the rod structures for these two rods are $(1,0)$ and $(0,1)$. Moreover, because the origin $p_{1}$ in the $\rho z$-plane corresponds to the vertex of the right-half quadrant under the inverse conformal map this is called a corner. For a general set of rod structures, a corner point $p_l$ is one which separates two axis rods, and a pole point is one which separates a horizon rod from an axis rod.

Potential constants $\mathbf{c}_{l}=(c_{l}^1,c_{l}^2)\in\mathbb{R}^2$ are prescribed on each axis rod $\Gamma_{l}$, and are used as boundary conditions for the twist potentials
$\omega_{i}|_{\Gamma_{l}}=c_{l}^{i}$. The constants may be chosen arbitrarily modulo the condition that they do not vary between adjacent rods separated by a corner. This is necessary for the construction of a model map in the next section, as well as a well-defined notion of angular momentum. In particular, the potential constants can only change after passing over a horizon rod, and this difference yields the angular momenta for each horizon component. Let $\mathcal{S}$ denote the 3-dimensional horizon cross section component associated with a horizon rod $\Gamma_{k}=[z_{k},z_{k-1}]$, then \eqref{chi}, \eqref{komar}, and \eqref{komar1} may be used to compute the Komar angular momenta of this component by
\begin{equation}
\mathcal{J}_{i}=\frac{1}{8\pi}\int_{\mathcal{S}}\star d\eta^{(i)}
=\frac{\pi}{2}\int_{\Gamma_{k}}\iota_{\eta^{(1)}}\iota_{\eta^{(2)}}\star d\eta^{(i)}
=\frac{\pi}{4}\int_{\Gamma_{k}}d\omega_{i}=\frac{\pi}{4}\left[\omega_{i}(p_{k-1})
-\omega_{i}(p_{k})\right].
\end{equation}
A rod data set $\mathcal{D}$ consists of the collection of corners and poles $\{p_l\}$,
rod structures $\{(m_l,n_l)\}$, and potential constants $\{\mathbf{c}_l\}$.

Consider now the topology of spacetime in a neighborhood of a corner point $p_l$ which separates axis rods $\Gamma_{l}$ and $\Gamma_{l+1}$ with rod structure $(m_l,n_l)$ and $(m_{l+1},n_{l+1})$. As is shown in the Appendix, new $2\pi$-periodic coordinates $(\bar{\phi}^{1},\bar{\phi}^2)$ may be chosen so that the rod structures with respect to these coordinates are given by $(1,0)$ and $(q,p)$, $p\neq 0$. That is, the Killing fields $\partial_{\bar{\phi}^1}$ and $q\partial_{\bar{\phi}^1}+p\partial_{\bar{\phi}^2}$ vanish on $\Gamma_{l}$ and $\Gamma_{l+1}$, respectively. Next take any semicircle in the $\rho z$-half plane (orbit space) centered at $p_l$ that connects a point on the interior of $\Gamma_l$ to a point on the interior of $\Gamma_{l+1}$. Note that each point on the interior of this semicircle represents a 2-torus in a constant time slice. By analyzing which 1-cycles collapse at the end points it follows that the semicircle represents a lens space $L(p,q)$. Recall that $L(1,q)\cong S^3$, so that when $p=\pm 1$ a neighborhood of the corner in a time slice is foliated by spheres, or rather a neighborhood of the corner in the spacetime is diffeomorphic to $\mathbb{R}^5$. It turns out that $p=\pm 1$ if and only if
\begin{equation} \label{admissibility1}
\det\begin{pmatrix} m_l & n_l \\ m_{l+1} & n_{l+1} \end{pmatrix} = \pm 1,
\end{equation}
and therefore the spacetime has trivial topology in a neighborhood of the corner if and only if the admissibility condition \eqref{admissibility1} holds, otherwise it has an orbifold singularity. The admissibility condition can be interpreted as stating that the
intersection number of the two 1-cycles that degenerate on either side of the corner is equal to $\pm 1$.

In addition to \eqref{admissibility1}, the main results of this paper rely on what will be referred to as the compatibility condition. This supplementary requirement is only valid when two consecutive corners are present. As described above, let $p_l$ be a corner separating axis rods $\Gamma_{l}$ and $\Gamma_{l+1}$, and suppose that there is another corner $p_{l-1}$ at the top end of $\Gamma_{l}$ connecting it to axis rod $\Gamma_{l-1}$. Assuming that the admissibility condition \eqref{admissibility1} holds at the two points $p_{l-1}$ and $p_{l}$, it may be arranged that these two determinants are $+1$ by multiplying each component of the rod structures by $-1$ if necessary. Observe that this operation on the rod structures does not change their properties, since the linear combinations of Killing fields that vanish at the rods is preserved.
The compatibility condition then states that the first component of the rod structures for $\Gamma_{l-1}$ and $\Gamma_{l+1}$ have opposite sign if both are nonzero
\begin{equation}\label{compatibilitycondition1}
m_{l-1}m_{l+1}\leq 0.
\end{equation}
This technical condition is used only in the construction of the model map in the next section. Unlike the admissibility condition, it is not known whether Theorem \ref{main} remains true without it. As mentioned in the introduction, if the admissibility condition is not assumed so that orbifold singularities are allowed then \eqref{compatibilitycondition1} should be enhanced to the generalized compatibility condition
\begin{equation}\label{gcompatibilitycondition}
m_{l-1}m_{l+1} \det\begin{pmatrix} m_{l-1} & n_{l-1} \\ m_{l} & n_{l} \end{pmatrix}
\det\begin{pmatrix} m_l & n_l \\ m_{l+1} & n_{l+1} \end{pmatrix} \leq 0.
\end{equation}
Note that the only way this quantity can vanish is if either $m_{l-1}=0$ or $m_{l+1}=0$, since for a corner the determinant is always nonzero.

Each connected component cross section of the event horizon has one of the following topologies \cite{HollandsYazadjiev1}: the sphere $S^3$, the ring $S^1\times S^2$, or a lens space $L(p,q)$. These manifolds have a singular foliation whose leaves are 2-dimensional tori, and whose singular leaves are circles resulting from the degeneration of a 1-cycle in the torus.  This can be observed geometrically from the canonical metric on each manifold as follows. The round metric on $S^3$ in Hopf coordinates is given by
\begin{equation}
d \theta^2 + \sin^2 \theta (d \phi^1)^2 + \cos^2 \theta (d \phi^2)^2,
\end{equation}
where $\theta \in [0, \pi/2]$, $\phi^i \in [0, 2 \pi]$.  For $0 < \theta < \pi/2$ the level set $\{\theta = \mbox{const.} \}$ is a flat 2-torus, and when $\theta = 0, \pi/2$ the level sets degenerates to $S^1$.  These singular leaves are characterized by the fact that the Killing fields $\partial_{\phi^1}$ and $\partial_{\phi^2}$ vanish at $\theta = 0, \pi/2$ respectively. Thus if $\theta$ is viewed as parameterizing a horizon rod, then the rod structure at the two poles (end points) is
$\{(1,0),(0,1)\}$. For the ring $S^1 \times S^2$ the canonical product metric is
\begin{equation}
[d \theta^2 + \sin^2 \theta (d \phi^1)^2] + (d \phi^2)^2,
\end{equation}
where $\theta \in [0, \pi]$, $\phi^i \in [0, 2 \pi]$.  The torus fibers are once again the level sets of $\theta$, and the singular leaves occur when $\theta = 0, \pi$ and coincide with the vanishing of the Killing field $\partial_{\phi^1}$, while the other Killing field $\partial_{\phi^2}$ never degenerates. The associated rod structure at the poles is then $\{(1,0),(1,0)\}$.

\begin{figure}
\includegraphics[width=5cm]{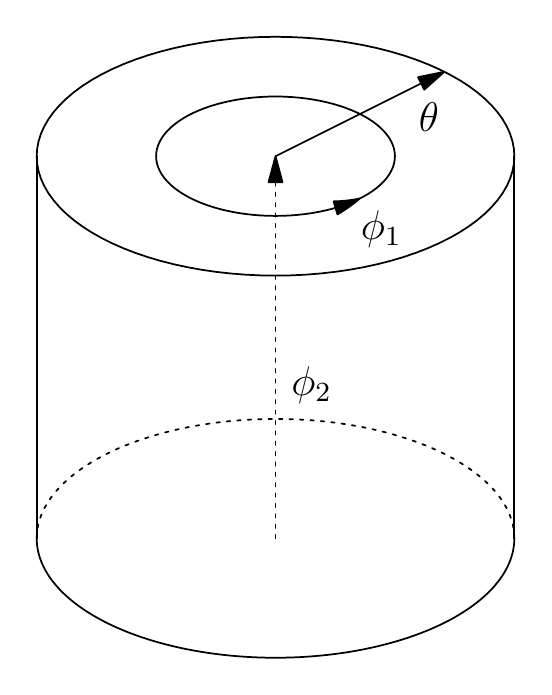}
\caption{Identification Space}  \label{toiletpaper}
\end{figure}

Consider now the lens space $L(p, q)=S^{3}/\mathbb{Z}_{p}$ which inherits its canonical metric
\begin{equation}
d \theta^2 + \sin^2 \theta (d \tilde{\phi}^1)^2 + \cos^2 \theta (d \tilde{\phi}^2)^2
\end{equation}
from the 3-sphere, where
\begin{equation}
\tilde{\phi}^1 = \phi^1 - \frac{q}{p} \phi^2, \quad \quad \tilde{\phi}^2 = \frac1p \phi^2,
\end{equation}
with $\theta \in [0, \pi/2]$, $\phi^i \in [0, 2 \pi]$.
Since $\phi^2$ has period $2\pi$, the following identifications are made
\begin{equation}
\tilde{\phi}^1 \sim \tilde{\phi}^1 + \frac{2\pi q}{p} , \quad\quad \tilde{\phi}^2 \sim  \tilde{\phi}^2 +  \frac{2\pi}{p} .
\end{equation}
The singular leaves at $\theta = 0, \pi/2$ are characterized by the vanishing of the Killing fields
\begin{equation}
\partial_{\tilde{\phi}^1}=\partial_{\phi^1},\quad\quad
\partial_{\tilde{\phi}^2}=q\partial_{\phi^1}+p\partial_{\phi^2},
\end{equation}
respectively, so that the associated rod structure at the poles is $\{(1,0),(q,p)\}$.
Recall the model of the lens space as a quotient space of the unit sphere $S^3=\{(z_1, z_2)\in\mathbb{C}^2 \,\, | \,\, |z_1|^2 + |z_2|^2 = 1 \}$ via the equivalence relation
\begin{equation}
(z_1, z_2 ) = (r_1 e^{ \tilde{\phi}^1 i}, r_2 e^{ \tilde{\phi}^2 i}) \sim (r_1 e^{\left(  \tilde{\phi}^1 + 2 \pi q/p\right)i}, r_2 e^{\left( \tilde{\phi}^2 + 2 \pi/p \right)i}).
\end{equation}
Here the pair of variables $(r_1, r_2)$ correspond to $(\sin \theta, \cos \theta)$ in the coordinates with which the lens space metric is written. A visualization of the lens space may be obtained by appropriately identifying the top, bottom, and sides of a solid cylinder as in Figure \ref{toiletpaper}. Namely, first collapse the external cylinder $\{\theta =\pi/2\}$ by identifying each vertical segment to a point, then identify the top and bottom discs via an orthogonal projection after performing a $2\pi q/p$ rotation of the top disc. The singular torus fibers occur where the action of the coordinate fields $\partial_{\tilde{\phi}^1}$ and $\partial_{\tilde{\phi}^2}$ degenerate, that is at $\theta=0,\pi/2$.

Using a similar analysis the topology of arbitrary rod structures may be understood. In Figure \ref{rod} four different rod structures for the orbit space are given, labeled by the topology of their horizons. Consider the first rod structure on the left in this diagram. The two semi-infinite rods are foliated by circle fibers none of which collapse, and hence they are 2-planes with an open disc removed. The finite rod has rod structure $(0,0)$ meaning that none of the rotational Killing fields vanish there. It is foliated by 2-tori such that each of the two 1-cycles generators in the torus degenerate on opposite poles. According to the description above, this yields an 3-sphere. Similarly, any simple curve in the $\rho z$-plane connecting the two semi-infinite rods also produces an $S^3$. In the second and third rod structures of Figure \ref{rod} it is clear that, by comparing with the singular foliations described above, these horizon rods represent a ring $S^1\times S^2$ and a lens $L(p,1)$, respectively. In these two examples there is also a different type of rod not present in the first example, namely a finite rod bounded by a pole on top and a corner on the bottom. This type of rod is foliated by circles with a singular leaf at the corner, and thus it gives a topological disc. The last example in Figure \ref{rod} has two horizon components in which the inner one is a lens $L(p,1)$ and the outer one is a ring $S^1\times S^2$, and hence the name `Black Lens Saturn'.

\begin{figure}
\includegraphics[width=12cm]{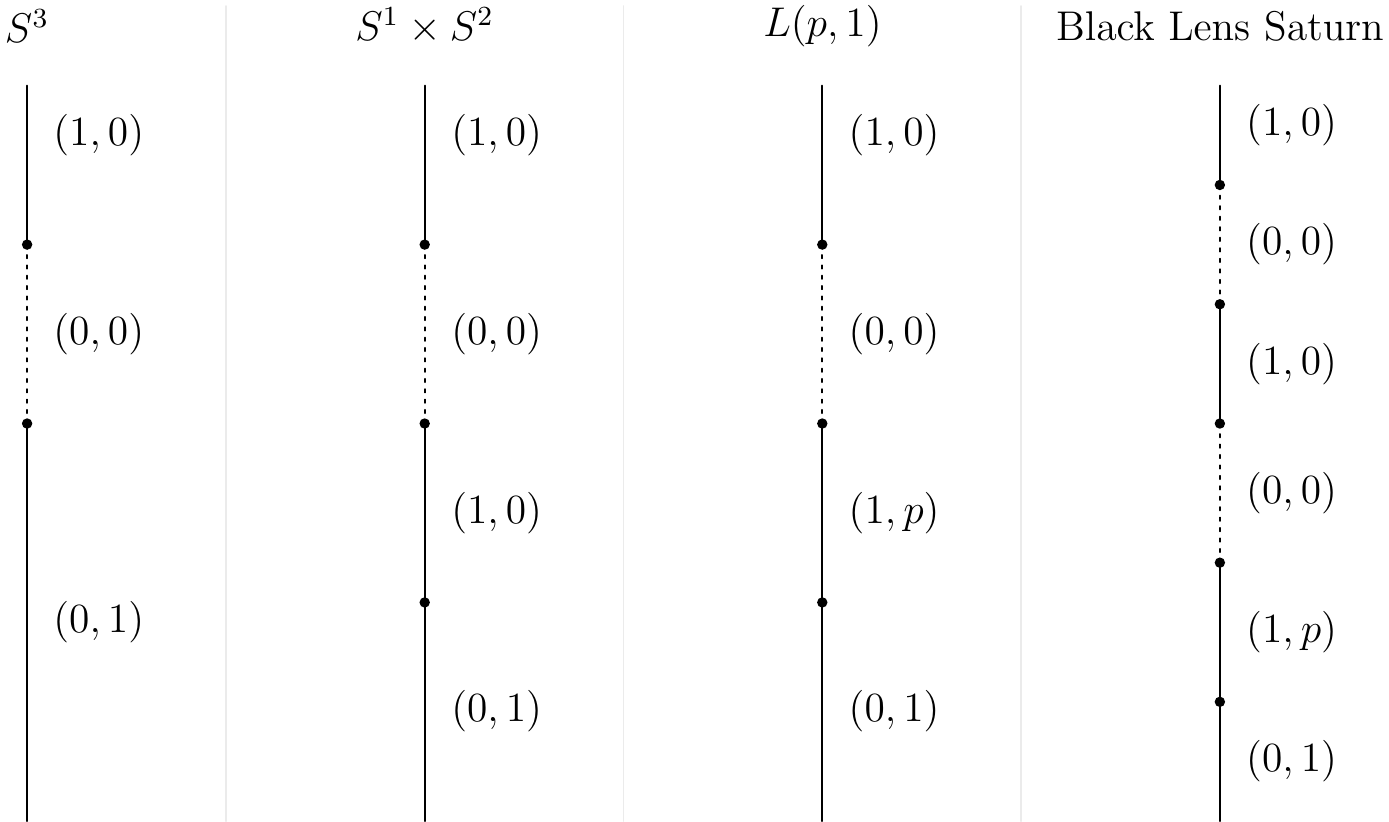}
\caption{Rod Strucures}  \label{rod}
\end{figure}

Observe that the rod structures of Figure \ref{rod} satisfy the admissibility condition \eqref{admissibility1} with $+1$ determinants, and the compatibility condition is vacuous. A natural question arises whether it is possible to produce a rod structure with a single horizon component having the general lens topology $L(p,q)$ without restricting to $q=1$, while at the same time satisfying the admissibility condition \eqref{admissibility1} and compatibility condition \eqref{compatibilitycondition1}. The following proposition answers this question affirmatively.

\begin{proposition}\label{lensrod}
Let $p$ and $q$ be integers satisfying $\mathrm{gcd}(p,q)=1$ and $p>q\geq 1$.
Then there exists a rod structure appropriate for an asymptotically flat spacetime of the form
\begin{equation}
\{(1, 0), (0, 0), (q, p),(q_1, p_1), \dots, (q_n, p_n), (0,\pm 1)\},
\end{equation}
which has a single lens space horizon $L(p, q)$, satisfies the admissibility condition \eqref{admissibility1} with positive determinants, and satisfies the compatibility condition \eqref{compatibilitycondition1}.
\end{proposition}

As an example observe that the single lens horizon $L(9,7)$ is realized by the rod structures
\begin{equation}
\{(1,0), (0,0), (7,9), (-4,-5), (-3, -4), (1, 1), (0, 1)\},
\end{equation}
which clearly satisfy the admissibility condition with positive determinants as well as the compatibility condition. In order to prove Proposition \ref{lensrod} we need a slightly modified version of Bezout's Lemma.

\begin{lemma}\label{Bezout}
Let $a \neq 1$ and $b \neq 1$ be relatively prime positive integers, then there exist integers $x$ and $y$ of the same sign such that
\begin{equation}
ax - by = 1,
\end{equation}
with $\mathrm{gcd}(x, y)=1$ and $1\leq |x| < b$, $1\leq |y| < a$. Furthermore, if $a<b$ then $|x| \geq |y|$.
\end{lemma}

\begin{proof}
By Bezout's Lemma there exist integers $\overline{x}$, $\overline{y}$ such that $a \overline{x} + b \overline{y} = 1$ with $|\overline{x}|\leq b$ and $|\overline{y}|\leq a$. Moreover, one of these may be an equality only if $a\mid b$ or $b\mid a$. Since $\mathrm{gcd}(a,b)=1$ it must hold that $|\overline{x}|< b$ and $|\overline{y}|< a$.
Furthermore, since $a, b>1$ we must have one of $\overline{x}$, $\overline{y}$ negative and the other positive.  Thus there are $\tilde{x}>0$, $\tilde{y}>0$ so that $a \tilde{x} - b \tilde{y} = \pm 1$, with $\tilde{x} < b$ and $\tilde{y} < a$.  If $\mathrm{gcd}(\tilde{x},\tilde{y})=c>1$ then $\tilde{x}=c\hat{x}$, $\tilde{y}=c\hat{y}$ and $c(a \hat{x} - b \hat{y}) = \pm 1$. This, however, is impossible since $c>1$, and hence $\mathrm{gcd}(\tilde{x},\tilde{y})=1$. If $a \tilde{x} - b \tilde{y} = 1$ then choose $(x, y)= (\tilde{x}, \tilde{y})$, and if  $a \tilde{x} - b \tilde{y} = -1$
then choose  $(x, y)=(-\tilde{x}, -\tilde{y})$. Lastly, neither $x$ nor $y$ may vanish as $a,b>1$.

Consider now the case when $a < b$. It then follows from the equation $ax-by=1$ that either $x>y$ (when $x,y>0$) or $x\leq y$ (when $x,y<0$). Hence $|x| \geq |y|$ when $a < b$.
\end{proof}

\begin{proof}[Proof of Proposition \ref{lensrod}]
If $q=1$ then append the rod structure $(0,1)$ after $(q,p)$ to solve the problem.
Assume now that $p$ and $q$ are relatively prime with $p>q>1$. Apply Bezout's Lemma with
$(a, b) = (q, p)$ to find a pair $(q_1, p_1)$ of relatively prime integers satisfying
\begin{equation}
qp_1-pq_1 =1
\end{equation}
as well as
\begin{equation}
 1 \leq |q_{1}| < q, \quad\quad\quad 1 \leq |p_{1}| < p.
\end{equation}
If  $|q_{1}|=1$, then by appending the rod structure $(0, \pm 1)$ after $(q_1, p_1) = (\pm1, p_{1})$ the desired result follows.

Consider now the case when $|q_1|> 1$. Again apply Bezout's Lemma to find $(\overline{q}_2, \overline{p}_2)$ relatively prime and satisfying
\begin{equation}
|q_1| \overline{p}_2- |p_1| \overline{q}_2 =1
\end{equation}
as well as
\begin{equation}
 1 \leq |\overline{q}_{2}| < |q_1|, \quad\quad\quad 1 \leq |\overline{p}_{2}| < |p_1|.
\end{equation}
Next define $(\tilde{q}_2, \tilde{p}_2)=\pm (\overline{q}_2, \overline{p}_2)$ where the sign is chosen so that
\begin{equation}
q_1 \tilde{p}_2- p_1 \tilde{q}_2 =1.
\end{equation}

The compatibility condition requires $q_0 q_2\leq 0$, and since $q_0=q>0$ this can be achieved by setting $(q_2,p_2)=(\tilde{q}_2,\tilde{p}_2)$ if $\tilde{q}_2<0$, and $(q_2,p_2)=(\tilde{q}_2-|q_1|,\tilde{p}_2-|p_1|)$ if $\tilde{q}_2>0$. Clearly this also satisfies the admissibility condition
\begin{equation}\label{DET}
q_1 p_2- p_1 q_2 =1
\end{equation}
as well as
\begin{equation}
 1 \leq |q_{2}| < |q_1|, \quad\quad\quad 1 \leq |p_{2}| < |p_1|,
\end{equation}
and \eqref{DET} implies that $q_2$ and $p_2$ are relatively prime.
Note that if it were the case that $q_0<0$ then $(|q_1|,|p_1|)$ should be added in the last step, rather than subtracted, in order to satisfy the compatibility condition.
This iterative process may be continued until $|q_n|=1$. Then at that point, append the rod structure $(0,\pm 1)$ after $(q_n,p_n)=(\pm 1,p_n)$ in order to achieve the stated outcome.
\end{proof}

We end this section by noting an important property of the horizon rods, which corresponds to a well-known result in 4-dimensional spacetime \cite{HawkingEllis}*{Proposition 9.3.1}.
Recall that a horizon rod is defined as an interval on the $z$-axis where the matrix $(f_{ij})$ is invertible, so that the torus fibers are nondegenerate there. These fibers together with the horizon rod form a codimension 2 surface in the spacetime, which will be referred to as a horizon rod surface.

\begin{lemma}
A horizon rod surface is a future apparent horizon, and within the $t=0$ slice it is a minimal surface.
\end{lemma}

\begin{proof}
At the beginning of this section we found that associated with a horizon rod there is a Killing field
\begin{equation}
\mathcal{K}=\partial_{t}+\Omega_1 \partial_{\phi^1}+\Omega_2 \partial_{\phi^2},\quad\quad\quad
\Omega_i\in\mathbb{R},
\end{equation}
which is null on the horizon rod surface $S$. Since the tangent space to $S$ is spanned by the vector fields $\partial_{z}$ and $\partial_{\phi^i}$, it easily follows from the structure of the spacetime metric \eqref{spacetimemetric} and the values for $\Omega_{i}$ that $\mathcal{K}$ is normal to $S$. The second fundamental form of $S$ in the $\mathcal{K}$-direction is then given by
\begin{equation}
II_{ab}= g(\nabla_{\partial_{a}}\mathcal{K},\partial_{b}),
\end{equation}
where $\partial_{a}$ denotes a tangent vector to $S$. Since $\mathcal{K}$ is Killing
\begin{equation}
g(\nabla_{\partial_{a}}\mathcal{K},\partial_{b})
=-g(\nabla_{\partial_{b}}\mathcal{K},\partial_{a}),
\end{equation}
and hence $II_{ab}$ is antisymmetric. Let
\begin{equation}
\gamma=f^{-1}e^{2\sigma}dz^2
+f_{ij}d\phi^{i}d\phi^{j}
\end{equation}
be the induced metric on the horizon rod surface, then the future null expansion is
\begin{equation}
\theta_{+}=\gamma^{ab}II_{ab}=0,
\end{equation}
since $\gamma^{ab}$ is symmetric.
By definition, $S$ is then a future apparent horizon.

In order to show that $S$ is minimal within the $t=0$ slice, let
\begin{equation}
\nu=(\nabla^a t)\partial_{a}=g^{tt}\partial_{t}+g^{t\phi^i}\partial_{\phi^i}
\end{equation}
be the unnormalized normal to the slice. Then the second fundamental form of the slice is
given by
\begin{equation}
|\nu|k_{cd}=g(\nabla_{\partial_{c}}\nu,\partial_d).
\end{equation}
Observe that
\begin{equation}
|\nu|k(\partial_{\phi^i},\partial_{\phi^j})
=g^{tt}g(\nabla_{\partial_{\phi^i}}\partial_{t},\partial_{\phi^j})
+g^{t\phi^l}g(\nabla_{\partial_{\phi^i}}\partial_{\phi^l},\partial_{\phi^j})
\end{equation}
is antisymmetric, and
\begin{equation}
|\nu|k(\partial_{z},\partial_{z})=g^{tt}g(\nabla_{\partial_{z}}\partial_{t},\partial_{z})
+g^{t\phi^l}g(\nabla_{\partial_{z}}\partial_{\phi^l},\partial_{z})=0,
\end{equation}
since $\partial_{t}$, $\partial_{\phi^i}$ are Killing. It follows that
\begin{equation}
\mathrm{Tr}_{S}k=\gamma^{ab}k_{ab}=f e^{-2\sigma}k(\partial_{z},\partial_{z})
+f^{ij}k(\partial_{\phi^i},\partial_{\phi^j})=0.
\end{equation}
Let $n$ denote the outward unit normal to $S$ within the $t=0$ slice, then
$n+\nu/|\nu|=\psi\mathcal{K}$ for some function $\psi$ on $S$. We then have
\begin{equation}
0=\psi\theta_{+}=H_{S}+\mathrm{Tr}_{S}k=H_{S}
\end{equation}
where $H_{S}$ denotes mean curvature, and therefore $S$ is a minimal surface within the slice.
\end{proof}

\section{The Model Map}
\label{sec5} \setcounter{equation}{0}
\setcounter{section}{5}

In this section a so called model map $\Phi_0\colon\R^3\setminus\Gamma\to \tilde{\mathbf{X}}\cong SL(3,\mathbb{R})/SO(3)$ is constructed, which encodes the prescribed asymptotic behavior near the axis and at infinity for the desired harmonic map, and also has finite tension. It may be viewed as an approximate solution to the singular harmonic map problem near the axes and at infinity.

The construction bears some similarity to the one in \cite{weinstein96}, but is more complex due to the abundance of rod structures, and the fact that even the non-rotating case is already nonlinear. We detail the construction in the case of a single component but the same approach works for all rod structures satisfying the compatibility condition. Where needed, we will point out differences required to make the approach work in the more general case.

The canonical Riemannian metric on $\tilde{\mathbf{X}}$ was constructed in Section \ref{sec3}, and it was noted that this space is parameterized by a $2\times 2$ symmetric positive definite matrix $F=(f_{ij})$ and a $2$-vector $\omega=(\omega_1,\omega_2)^t$. If $f=\det F$ then the metric in these coordinates \cite{IdaIshibashiShiromizu} is given by
\begin{align}
\begin{split}
	\tilde{\mathbf{g}} =& \frac14 \frac{df^2}{f^2} +\frac14 f^{ij}f^{kl}df_{ik}df_{jl}  + \frac12 \frac{f^{ij} d\omega_i d\omega_j}{f} \\
	=& \frac14 [\mathrm{Tr}(F^{-1}dF)]^2  +\frac14 \mathrm{Tr}(F^{-1}dF\, F^{-1} dF)  + \frac12 \frac{d\omega^t \, F^{-1}\, d\omega}{f}.
\end{split}
\end{align}
A computation shows that the components of the tension \eqref{tensiondef} of a map $\Phi_0=(F,\omega)$ are
\begin{align}\label{eulerlagrange}
\begin{split}
\tau^{f_{lj}}=&\Delta f_{lj}-f^{km}\nabla^{\mu}f_{lm}\nabla_{\mu}f_{kj}
+f^{-1}\nabla^{\mu}\omega_{l}\nabla_{\mu}\omega_{j},\\
\tau^{\omega_{j}}=&\Delta\omega_{j}-f^{kl}\nabla^{\mu}f_{jl}\nabla_{\mu}\omega_{k}
-f^{lm}\nabla^{\mu}f_{lm}\nabla_{\mu}\omega_{j},
\end{split}
\end{align}
where $\Delta$ is the Laplacian and $\nabla$ the connection associated with the flat metric \eqref{flatmetric} on $\mathbb{R}^3$. This yields the harmonic map equations $\tau=0$ in these coordinates. Let
\begin{equation}
H=F^{-1}\nabla F,\quad\quad G=f^{-1}F^{-1}\left(\nabla\omega\right)^2,\quad\quad
K=f^{-1}F^{-1}\nabla\omega,
\end{equation}
that is
\begin{equation}
  H_{\mu}{}^i{}_j= f^{ik}\nabla_\mu f_{kj},\quad\quad
  G^{i}_j=f^{-1} f^{ik} \nabla_\mu\omega_k\,\nabla^\mu\omega_j,\quad\quad
  K_\mu{}^i=f^{-1}f^{ij}\nabla_\mu\omega_j,
\end{equation}
and observe that
\begin{equation}
\left(\operatorname{div}H+G\right)^{i}_j=f^{il}\tau^{f_{lj}},\quad\quad
\left(\operatorname{div}K\right)^{i}=f^{-1}f^{ij}\tau^{\omega_{j}}.
\end{equation}
We then have
\begin{equation}
  |\tau|^2=\frac14 \left[ \mathrm{Tr}(\operatorname{div}H + G)\right]^2 + \frac14 \mathrm{Tr} \left[(\operatorname{div}H + G)(\operatorname{div}H + G) \right]
  + \frac12 f (\operatorname{div}K)^t F (\operatorname{div}K).
\end{equation}

In order to state the main result of this section we will say that a map $\Phi_0=(F,\omega)$ \textit{respects} a rod data set $\mathcal{D}$, if $(m_l,n_l)$ is the rod structure and $\mathbf{c}_l$ the potential constant within $\mathcal{D}$ for an axis rod $\Gamma_l$ then
\begin{equation}
(m_l,n_l)\in\mathrm{ker \text{ }} F|_{\Gamma_l},\quad\quad\quad \omega|_{\Gamma_l}=\mathbf{c}_l.
\end{equation}

\begin{theorem} \label{model}
Given a rod data set $\mathcal{D}$ satisfying the generalized compatibility condition
\eqref{gcompatibilitycondition}, there exists a model map $\Phi_0\colon\R^3\setminus\Gamma\to \tilde{X}$ with uniformly bounded tension having decay $|\tau|=O(r^{-7/2})$ which respects $\mathcal{D}$.
\end{theorem}

\begin{proof}
As mentioned above, we give a detailed proof for the case of the rod configuration corresponding to a single lens horizon $L(p,1)$, see Figure~\ref{domain}. However, we will indicate below the changes required for the general case.

\begin{figure}
\includegraphics[width=10cm]{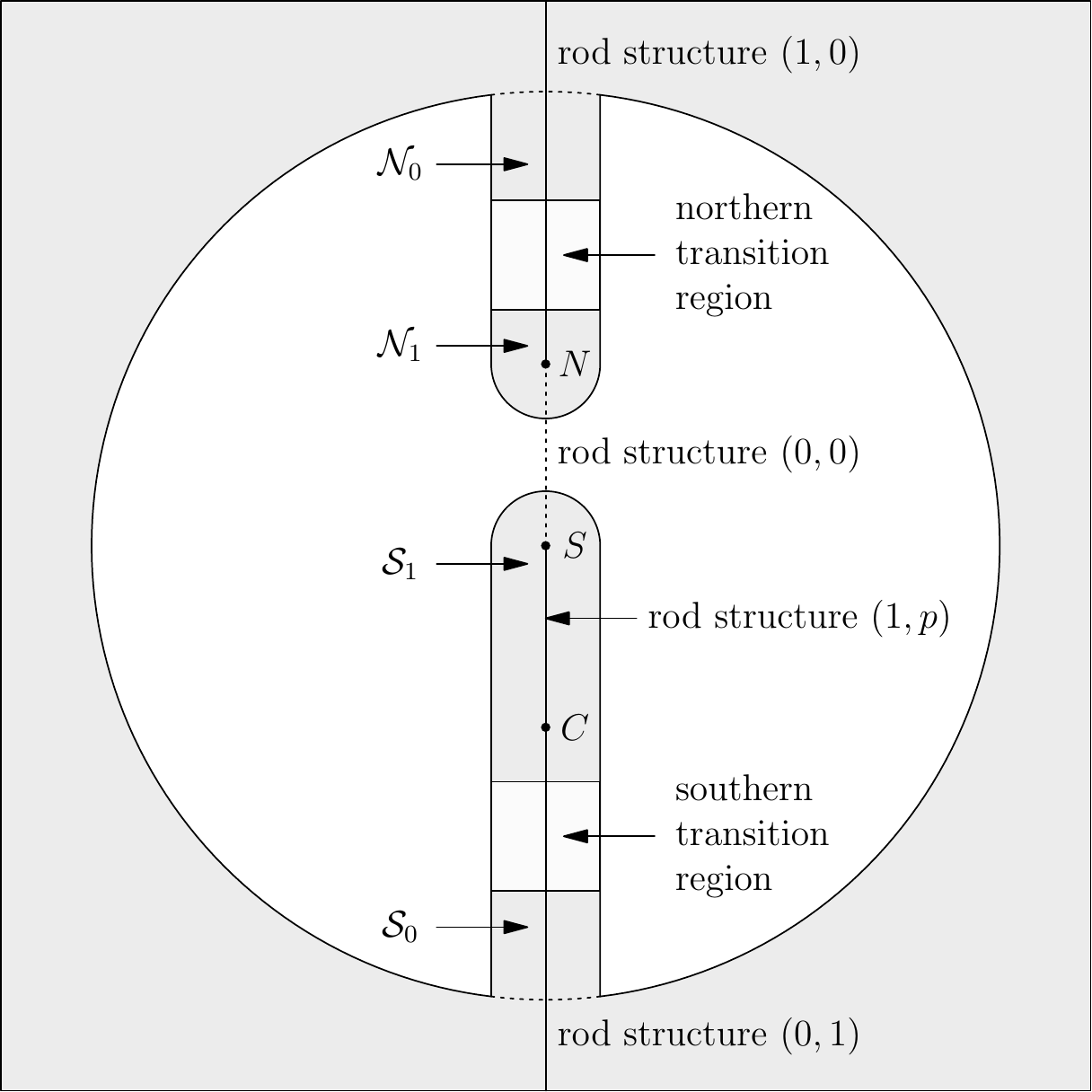}
\caption{Model Map Construction}  \label{domain}
\end{figure}

The only requirement of the map $\Phi_0$ within the white area in Figure~\ref{domain} will be that it is a smooth extension of the map which will be defined explicitly in the gray region. This can easily be achieved since the white area remains a fixed distance away from the singular set $\Gamma$, and this clearly implies that the tension is bounded within the white area.

For convenience, we define a pair of harmonic functions needed in the construction. For $a\in\R$ let $r_{a}$ be the Euclidean distance from the point $z=a$ on the $z$-axis, and let $\theta_a$ be the polar angle about this center. Then set
\begin{equation}
	u_a =\log(r_a-(z-a))=\log\bigl(2r_a\sin^2(\theta_a/2)\bigr), \quad\quad
v_a = \log(r_a+(z-a))=\log\bigl(2r_a\cos^2(\theta_a/2)\bigr).
\end{equation}
It is easy to check that these functions are harmonic. Furthermore $u_a$ behaves like $2\log\rho$ near the $z>a$ part of the $z$-axis and is locally bounded below on the $z<a$ part of the $z$-axis. Also, clearly $u_a(\rho,z-a) = v_a(\rho,-(z-a))$ and hence $v_a$ behaves like $2\log\rho$ on the $z<a$ part of the $z$-axis and is locally bounded below on the $z>a$ part of the $z$-axis.

We begin with the definition of $\Phi_0$ outside a large ball. The map there is based on the Minkowski metric \eqref{Minkowski} and is given by
\begin{equation}
	F=\begin{pmatrix} e^{u_0-\log2} & 0 \\ 0 & e^{v_0-\log2} \end{pmatrix}, \qquad \omega=\omega(\theta),
\end{equation}
where $\theta=\theta_0$. The function $\omega(\theta)$ is smooth and chosen so that $\omega$ is the appropriate constant on $[0,\epsilon]\cup[\pi-\epsilon,\pi]$, with $0<\epsilon<\pi/2$ fixed so that $\omega$ is constant on the regions $\mathcal N_0$ and $\mathcal S_0$. Observe that this map is harmonic wherever $\omega$ is constant, since $G=0$ and $\div F^{-1}\nabla F=0$. It will now be shown that the tension $|\tau|$ decays like $O(r^{-7/2})$, which as will be seen later is sufficient for the main existence and uniqueness arguments. Since the tension vanishes for $\theta\in[0,\epsilon]\cup[\pi-\epsilon,\pi]$, we need only estimate $|\tau|$ on the interval $[\epsilon,\pi-\epsilon]$. An explicit calculation gives
\begin{align}
\begin{split}
	f (\operatorname{div}K)^t F (\operatorname{div}K) =&
    \frac{4\csc^2\theta \sin^2(\theta/2)}{r^7}
    \left[\csc^4(\theta/2) \bigl(\omega_1''-(\csc\theta+2\cot\theta) \omega_1'\bigr)^2 \right. \\
    &\left. + 4 \csc^2\theta \bigl(\omega_2''+(\csc\theta-2\cot\theta) \omega_2'\bigr)^2\right]\\
    =& O(r^{-7}),
\end{split}
\end{align}
and
\begin{equation}
	G=\frac{\csc^2(\theta/2)\sec^2(\theta/2)}{r^5}
	\begin{pmatrix}
		 \omega _1'^2\csc^2(\theta/2) &
		 \omega _1'\omega _2\csc^2(\theta/2) \\[1ex]
		 \omega _1'\omega _2 \sec^2(\theta/2) &
		 \omega _2'^2 \sec^2(\theta/2) &
	\end{pmatrix}
	=O(r^{-5}).
\end{equation}
Since $\div H=0$, it follows that $|\tau|=O(r^{-7/2})$.

It remains to define the map inside the two tubular neighborhoods capped with hemispheres. Consider first the northern tubular neighborhood. Let $z=b$ indicate the location of the point $N$. Then in this region define
\begin{equation}
	F=\begin{pmatrix} e^{u} & 0 \\ 0 & e^{v} \end{pmatrix}, \qquad \omega=\mathbf{c}_1,
\end{equation}
where
\begin{equation}
u=\lambda (u_0-\log2) + (1-\lambda) u_b,\quad\quad v=\lambda (v_0-\log2),
\end{equation}
and $\lambda=\lambda(z)$ is a smooth cut-off function with $\lambda=1$ in $\mathcal N_0$ and $\lambda=0$ in $\mathcal N_1$. This leads to the correct rod structure, and the definitions outside the large ball and in $\mathcal N_0$ agree. Moreover
\begin{equation}
	\div H = \begin{pmatrix}  \Delta [\lambda(u_0-u_b)] & 0 \\ 0 & \Delta [\lambda v_0]\end{pmatrix},
\end{equation}
which is bounded. Indeed
\begin{equation}
	\Delta[\lambda(u_0-u_b)] = (u_0-u_b)\Delta\lambda + 2(\partial_{z}\lambda) \partial_z(u_0-u_b),
\end{equation}
and $\partial_z u_a = 1/r_a$ (on the $z$-axis) for $a=0, b$
is clearly bounded in the transition region. Similarly $\Delta[\lambda v_0]$ is bounded since $\partial_z v_0 = -1/r_0=-1/r$ (on the $z$-axis) is bounded. It follows that $|\tau|$ is bounded in the northern region, as $G=0$ and $K=0$ due to the constancy of $\omega$.

Consider now the southern tubular neighborhood. The map in $\mathcal S_0$ is defined exactly as in $\mathcal N_0$, that is with the same $F$ but with $\omega=\mathbf{c}_2$. In fact $\omega$ is set to be the constant $\mathbf{c}_2$ in the entire southern tubular neighborhood. Next, let the south pole $S$ and corner point $C$ be located at $z=c$ and $z=0$,
respectively. Then in $\mathcal S_1$ the remainder of the map is defined by
\begin{equation}
	F= hF_0h^t = h \begin{pmatrix} e^{u} & 0 \\ 0 & e^{v} \end{pmatrix}h^t,
\end{equation}
where
\begin{equation} \label{h}
	h = \begin{pmatrix}  1& -p \\ 0 & 1 \end{pmatrix}
\end{equation}
and $v=v_0-\log2$, $u=u_0-u_c$. As before $\div(F_0^{-1}\nabla F_0)=0$ and hence \begin{equation}
\div(F^{-1}\nabla F) = h^{-t} \div(F_0^{-1}\nabla F_0) h^t = 0,
\end{equation}
where for notational convenience $h^{-t}:=(h^t)^{-1}$.
It follows that $\Phi_0$ is a harmonic map in $\mathcal{S}_1$. In order to verify
that the rod structure is correct, observe that
\begin{equation} \label{structure}
	F\vector{1}{0} = \vector{e^u+p^2e^v}{-pe^v}, \quad\quad F\vector01 = \vector{-pe^v}{e^v}, \quad\quad F\vector1p = \vector{e^u}{0}.
\end{equation}
From this it is clear that the only direction which degenerates on the disk rod (between $S$ and $C$) is $(1,p)$, and the only direction that degenerates on the south rod (below $C$) is $(0,1)$. Furthermore, since $F_0$ is nonsingular on the horizon rod the same is true of $F$.

Lastly, the map will be defined on the southern transition region. Recall that $\omega$ is constant. Moreover if $F$ defined in $\mathcal{S}_1$ can be transitioned
to a diagonal $F$ satisfying $\div(F^{-1}\nabla F)=0$, then we can complete the transition in the same manner as in the northern transition region. Thus it remains to demonstrate the transition to a diagonal $F$. Set
\begin{equation}
	F =  h(z) F_0 h(z)^t, \qquad\quad  h(z) = \begin{pmatrix}  1& -p\lambda(z) \\ 0 & 1 \end{pmatrix},
\end{equation}
where $F_0$ is as above, and $\lambda(z)$ is a smooth cut-off function which is equal to $1$ near $\mathcal S_1$ and equal to $0$ near $\mathcal S_0$. To verify that $\div(F^{-1}\nabla F)$ is bounded in the transition region compute
\begin{equation}
	F^{-1} \nabla F
	= (F_0h^t)^{-1} (h^{-1}\nabla h) F_0h^t + h^{-t}  (F_0^{-1}\nabla F_0) h^t + h^{-t}\nabla h,
\end{equation}
and
\begin{align}
\begin{split}
\label{divFdF}
	\div(F^{-1} \nabla F) = &[\nabla(F_0h^t)^{-1}]\cdot (h^{-1}\nabla h) F_0h^t + (F_0h^t)^{-1} \div (h^{-1}\nabla h) F_0h^t\\
	&+ (F_0h^t)^{-1} (h^{-1}\nabla h)\cdot \nabla (F_0h^t)
+ (\nabla h^{-t}) \cdot (F_0^{-1}\nabla F_0) h^t\\
	&+h^{-t}  \div(F_0^{-1}\nabla F_0) h^t + h^{-t}  (F_0^{-1}\nabla F_0)\cdot \nabla h^t + \div(h^{-t}\nabla h).
\end{split}
\end{align}
Each term may now be estimated individually. First note that
the fifth term vanishes and the seventh term is clearly bounded.
Furthermore
\begin{equation}
	F_0^{-1} \nabla F_0 = \begin{pmatrix} \nabla u & 0 \\ 0 & \nabla v \end{pmatrix},
\end{equation}
and since $h$ depends only on $z$ we may replace $\nabla u$ and $\nabla v$ in \eqref{divFdF} by $\partial_z u$ and $\partial_z v$, respectively. As explained above these $z$-derivatives are bounded, and since $h^t$, $h^{-t}$, $\partial_{z}h^t$ and $\partial_{z}h^{-t}$ are bounded it follows that the fourth and sixth terms are bounded. Next observe that the second term becomes
\begin{equation}
	 (F_0h^t)^{-1} \div (h^{-1}\nabla h) F_0h^t= pe^{v-u}\lambda''
	 \begin{pmatrix} p\lambda & -1 \\ p^2\lambda^2 & -p\lambda \end{pmatrix},
\end{equation}
which is bounded. Furthermore the sum of the first and third terms is
\begin{align}
\begin{split}
	&[\nabla(F_0h^t)^{-1}]\cdot (h^{-1}\nabla h) F_0h^t
	+ (F_0h^t)^{-1} (h^{-1}\nabla h)\cdot \nabla (F_0h^t)\\
 =&
	pe^{v-u}\lambda'
	\begin{pmatrix} p\bigl[\lambda (\partial_z v-\partial_z u)+\lambda'\bigr] & \partial_z u-\partial_z v \\
	p^2\lambda\bigl[\lambda(\partial_z v-\partial_z u)+2\lambda'\bigr] & -p\bigl[\lambda(\partial_z v-\partial_z u)+\lambda'\bigr] \end{pmatrix},
\end{split}
\end{align}
which again is bounded. It follows that $|\tau|$ is bounded in the southern region, and this completes the proof for the rod data set associated with a single component lens horizon $L(p,1)$.

\begin{remark} \label{integer}
We note that in the argument above showing that $\div(F^{-1}\nabla F)$ is bounded no use was made of the fact that $p$ is an integer. This is will be important in what follows.
\end{remark}

Consider now the case of a general rod data set, in which consecutive corners may be present. In this situation the map will be defined inductively one corner at a time, with a transition region between any two consecutive corners, as well as a transition region on each of the two semi-infinite rods.  The only feature which remains to be treated is the case of two consecutive corners. Suppose then that consecutive corners occur at points $C_N$ and $C_S$ along the $z$-axis, with $z=\mathrm{a}$ and $z=\mathrm{b}$ at $C_N$ and $C_S$ respectively. Let there be rod structures $(m,n)$ above $C_N$, $(p,q)$ between $C_N$ and $C_S$, and $(r,s)$ below $C_S$. It will be assumed that $m\neq 0$, $p\neq 0$, $r\neq 0$, and that the generalized compatibility condition is satisfied
\begin{equation}\label{1256}
mr(ps-rq)(mq-np)\leq 0.
\end{equation}
Note that this quantity is nonzero (and hence negative) since $ps-rq\neq 0$ and $mq-np\neq 0$ due to the fact that $C_N$ and $C_S$ are genuine corners.

Let $v=u_{\mathrm{b}}-u_{\mathrm{a}}$ and $u=2\log\rho-v$ and set
\begin{equation}
	F_0=\begin{pmatrix} e^u & 0 \\ 0 & e^v \end{pmatrix},
\end{equation}
so that $F_0$ gives rod structure $(1,0)$ above $C_N$ and below $C_S$, and $(0,1)$ between $C_N$ and $C_S$. Next define $F_N=h_NF_0h_N^t$ near $C_N$ and $F_S=h_SF_0h_S^t$ near $C_S$, where
\begin{equation} \label{hnhs}
	h_N=\begin{pmatrix} -q/p & -n/m \\ 1 & 1 \end{pmatrix}, \qquad
	h_S=\begin{pmatrix} -q/p & -s/r \\ 1 & 1 \end{pmatrix}.
\end{equation}
It is straightforward to check that the maps $F_N$ and $F_S$ yield the desired rod structures on each of the three rods in neighborhoods of $C_N$ and $C_S$ respectively, and that $(F_N,\omega)$ and $(F_S,\omega)$ are harmonic whenever $\omega$ is constant. This latter property arises from the fact that although $F\mapsto hFh^t$, $\omega\mapsto h\omega$ is an isometry of $\tilde{\mathbf{X}}$ if and only if $\det h=\pm1$, this determinant condition is not required here for the harmonic map equations to be satisfied since $\omega$ is constant. It remains to define $F$ in a transition region between $C_N$ and $C_S$. In order to do this first let $\bar{F}_N=\mathbf{k} F_0 \mathbf{k}^t$ and $\bar{F}_S=F_0$, where
\begin{equation}\label{987}
	\mathbf{k} = h_S^{-1}h_N=\begin{pmatrix} 1 & \frac{p(ms-nr)}{m(ps-qr)} \\[1ex]
           0 & -\frac{r(mq-np)}{m(ps-qr)} \end{pmatrix}.
\end{equation}
If there is a smooth transition $\mathbf{k}=\mathbf{k}(z)$ from $h_S^{-1}h_N$ to
\begin{equation}\label{jfh}
	 \begin{pmatrix} 1 & \frac{p(ms-nr)}{m(ps-qr)} \\[1ex] 0 & 1 \end{pmatrix},
\end{equation}
then by Remark~\ref{integer}
it is clear that we can further transition $\mathbf{k}$ to the identity as in the arguments above the remark, since the only difference between \eqref{jfh} and $h$ in \eqref{h} is the fact that the off-diagonal element is an integer in the latter matrix.
It follows that $\bar{F}$ would then be defined in the whole region encompassing both corners, having the property that it is equal to $\bar{F}_N$ near $C_N$ and equal to $\bar{F}_S$ near $C_S$. Finally, taking $F=h_S \bar{F} h_S^t$ produces a map with finite tension which coincides with $F_N$ near $C_N$ and $F_S$ near $C_S$.

It remains to define the transition from \eqref{987} to \eqref{jfh}. Set
\begin{equation}
	\mathbf{k}(z) = \begin{pmatrix} 1 & \varsigma \\ 0 & \lambda(z) \end{pmatrix},\quad\quad\quad \varsigma=\frac{p(ms-nr)}{m(ps-qr)},
\end{equation}
where $\lambda(z)$ is a smooth cut-off function satisfying $\lambda(z)=-\frac{r(mq-np)}{m(ps-qr)}$ near $C_N$ and $\lambda(z)=1$ for $z<(\mathrm{a}+\mathrm{b})/2$.
According to the generalized compatibility condition \eqref{1256},
$\lambda(z)$ may be chosen strictly positive. The arguments following \eqref{divFdF} may now be repeated to show that the tension remains bounded. In particular, the terms four through seven of \eqref{divFdF} are bounded in the current setting. By denoting $F_{\mathbf{k}}=F_0 \mathbf{k}^t$ the second term becomes
\begin{equation}
	F_{\mathbf{k}}^{-1} \div (\mathbf{k}^{-1}\nabla \mathbf{k})F_{\mathbf{k}} =  \varsigma e^{v-u}\lambda'
	\begin{pmatrix} -\dfrac{\varsigma}{\lambda} & -1 \\[1.5ex]
	\dfrac{\varsigma^2+e^{u-v}}{\lambda^2} & \dfrac{\varsigma^2e^{u-v}}{\varsigma\lambda}
	\end{pmatrix},
\end{equation}
and the sum of the first and third terms is
\begin{equation}
	\nabla F_{\mathbf{k}}^{-1} \cdot(\mathbf{k}^{-1}\nabla \mathbf{k}) F_{\mathbf{k}} + F_{\mathbf{k}}^{-1} (\mathbf{k}^{-1}\nabla \mathbf{k})\cdot \nabla F_{\mathbf{k}} = \varsigma e^{v-u}\lambda'
	\begin{pmatrix} \dfrac{\varsigma(u_z-v_z)}{\lambda} & u_z - v_z - \dfrac{\lambda'}{\lambda} \\[1.5ex]
	\dfrac{\varsigma^2\lambda(u_z-v_z)+(\varsigma^2+e^{u-v}\lambda')}{\lambda^3} & \dfrac{\varsigma(v_z-u_z)}{\lambda}
	\end{pmatrix},
\end{equation}
both of which are bounded. Similar arguments may be used to treat the cases when one of $m$, $p$, $r$ is zero.
\end{proof}

\section{Energy Estimates}
\label{sec6} \setcounter{equation}{0}
\setcounter{section}{6}

In the rank 1 case treated in \cite{weinstein96}, a priori estimates for the singular harmonic map problem relied heavily on the uniformly strict negative curvature of the target spaces. In the current setting the target symmetric space $\mathbf{X}=SL(3,\mathbb{R})/SO(3)$ is of rank 2, that is the dimension of a maximal flat subspace is 2. It follows that $X$ is of nonpositive curvature and the methods of \cite{weinstein96} break down. In order to overcome this difficulty, we will employ a generalization of horospherical coordinates from hyperbolic space so that the flat directions as well as the coordinate planes of strict negative curvature are explicitly identified, and are thus more easily exploited. Coordinate systems of the symmetric space $\mathbf{X}=SL(3,\mathbb{R})/SO(3)$ have been investigated previously, as in \cite{MazzeoVasy}, yet what we need requires a different set of properties.

Consider the Iwasawa decomposition \cite{BallmanGromovSchroeder} of $G=SL(3,\mathbb{R})$ given by $G=KAN$ where the three subgroups are $K=SO(3)$,
\begin{equation}
A = \{ \mathrm{diag} (\lambda_1, \lambda_2, \lambda_3) \,\, | \,\, \lambda_i > 0,\text{ for }i=1,2,3,\text{ } \lambda_1 \lambda_2 \lambda_3 =1 \},
\end{equation}
and
\begin{equation}
N=\{\text{upper triangular matrices with 1's on the diagonal}\}.
\end{equation}
For each $g\in G$ there exist unique elements $k\in K$, $a\in A$, and $n\in N$ such that
$g=kan$. Moreover by taking inverses we have $G=NAK$, and hence $\mathbf{X}=G/K$ may be identified with the subgroup $NA$. Let $x_0=[Id]\in \mathbf{X}$ then the orbit $A\cdot x_0$ represents a maximal flat so that it is a totally geodesic submanifold with vanishing curvature. The last property follows from the curvature formula in Section \ref{sec3}, and the fact that the Lie algebra
\begin{equation}
\mathfrak{a} = \{ \mathrm{diag} (\lambda_1, \lambda_2, \lambda_3) \,\, | \,\, \sum \lambda_i = 0\}
\end{equation}
associated with $A$ is abelian ie. $[\alpha_1,\alpha_2]=0$ for all $\alpha_1, \alpha_2\in \mathfrak{a}$.
On the other hand, the orbit $N\cdot x_0$ is a horocycle determined by the Weyl chamber
\begin{equation}
\mathfrak{a}^+ = \{ \mathrm{diag} (\lambda_1, \lambda_2, \lambda_3) \,\, | \,\,
\lambda_1>\lambda_2>\lambda_3,\text{ }\sum \lambda_i = 0\}\subset\mathfrak{a}.
\end{equation}
It is a closed submanifold with the property that every flat which is asymptotic to the Weyl chamber at infinity
\begin{equation}
\mathrm{w}^{+}:=(A^{+}\cdot x_0)(\infty)=\{\gamma(\infty)\mid \gamma(s)=\mathrm{exp}(s\alpha^{+})\cdot x_0,\text{ }\alpha^{+}\in\mathfrak{a}^{+}\},
\end{equation}
intersects the horocycle orthogonally in exactly one point; recall that a flat $\mathcal{F}$ is asymptotic to a Weyl chamber $\mathrm{w}$ at infinity if $\mathrm{w}\subset\mathcal{F}(\infty)$.
In particular, the horocycle $N \cdot x_0$ and flat $\mathcal{F}_{x_0}:=A \cdot x_0$
intersect orthogonally at $x_0$, as can be seen from the orthogonality between the respective Lie algebras $\mathfrak{n}$ (all upper triangular matrices with zeros on the diagonal) and $\mathfrak{a}$ with respect to the Riemannian metric at $x_0$ given in Section \ref{sec3}.

A foliation by flats may be constructed \cite{BallmanGromovSchroeder} from the action of $N$. More precisely
\begin{equation}
\mathbf{X} = \bigcup_{n \in N} n \cdot \mathcal{F}_{x_0},
\end{equation}
where $n\cdot \mathcal{F}_{x_0} \cap n' \cdot \mathcal{F}_{x_0} = \emptyset$ for $n \neq n'$ and each $n \cdot\mathcal{F}_{x_0}$ is asymptotic to the Weyl chamber $\mathrm{w}^+$. Since each point $x \in \mathbf{X}$ can be uniquely written as
$n a \cdot x_0$, and $a \cdot \mathcal{F}_{x_0} = \mathcal{F}_{x_0}$ as sets,
the assignment $x \mapsto \mathcal{F}_x=na\cdot\mathcal{F}_{x_0}$ defines a smooth foliation of $\mathbf{X}$ whose leaves are the set of totally geodesic submanifolds $\{ n \cdot F_{x_0}\}_{n \in N}$, each of which is isometric to ${\mathbb R}^2$. By homogeneity of $\mathbf{X} = G/K$,  the 3-dimensional horocycle $N \cdot x$ and the 2-dimensional flat $\mathcal{F}_x$ intersect orthogonally at (and only at) $x$.  In this sense, the pair $(a, n)$ gives a horocyclic orthogonal coordinate system for $\mathbf{X}$.

Let $\gamma_{x_0}(s)$ be an arc-length parameterized geodesic satisfying $\gamma_{x_0}(0)  = x_0$, and $\gamma_{x_0}(\infty) \in \mathrm{w}^+$. Equivalently $\gamma_{x_0}'(0) \in T_{x_0} X$ is an element of a Weyl chamber $\mathfrak{a}^+$, so that $\gamma_{x_0}$ is regular in the sense that it is contained in a unique 2-dimensional flat, namely $\mathcal{F}_{x_0}$. Since the action by $na$ on $\mathbf{X}$ is isometric and preserves the combinatorial structure of the Weyl chambers projected to $\mathbf{X}(\infty)$, it follows that $\gamma_{x}(s):=na\cdot \gamma_{x_0}(s)$ is a regular geodesic contained in the flat $n\cdot\mathcal{F}_{x_0}$, and is asymptotic to $\mathrm{w}^+$. In fact, the distance $d_{\mathbf{X}}(n\cdot \gamma_{x_0}(s),\gamma_{x_0}(s))$ decays exponentially and $d_{\mathbf{X}}(na\cdot \gamma_{x_0}(s),\gamma_{x_0}(s))\rightarrow d_{\mathbf{X}}(a\cdot x_0, x_0)$.

On the flat $\mathcal{F}_{x_0}$ there is a natural Euclidean coordinate system $r=(r_1, r_2)$, where the origin is identified with $x_0$, the $r_1$-axis coincides with the regular geodesic $\gamma_{x_0}(s)$, and the $r_2$-axis is the orthogonal line to $\gamma_{x_0}(s)$.  The $r_1$ axis is chosen to have the opposite orientation from that of $\gamma_{x_0}$, so that $r_1\rightarrow\infty$ corresponds to $s\rightarrow -\infty$, and similarly for $r_2$.
The $(r_1, r_2)$ coordinate system may then be pushed forward to the flat $n \cdot \mathcal{F}_{x_0}$ where the origin is identified with $n \cdot x_0$, the $r_1$-axis is the geodesic $\gamma_{n \cdot x_0}(s)$, and the $r_2$-axis is again the orthogonal line to $\gamma_{n x_0}(s)$ in the flat.  Hence the the horocyclic coordinates $(a, n)$ may be represented by $(r, n)$. Moreover, for each $n'\in N$ there is an isometry which preserves the $r$-coordinates and for each $r'$ there is a diffeomorphism which preserves the $n$-coordinates
\begin{equation}
\Xi_{n'} : (r_1, r_2, n) \mapsto (r_1, r_2, n' n),\quad\quad
\Xi_{r'}: (r_1, r_2, n) \mapsto (r_1+r_{1}', r_2+r_{2}', n).
\end{equation}
The $r$-translations map horocycles to horocylces, and thus if
$\theta=(\theta^1,\theta^2,\theta^3)$ is a system of global coordinates on $N\cdot x_0\cong\mathbb{R}^3$ then they may be pushed forward to all horocycles by the action of $\Xi_{r'}$.
It follows that $(r,\theta)$ form a system of global coordinates on $\mathbf{X}$ with the property that the coordinate fields $\partial_{r_i}$ and $\partial_{\theta^j}$ are orthogonal. By combining the observations above, the $G$-invariant Riemannian metric on $\mathbf{X}$ can be expressed in these coordinates by
\begin{equation}
\mathbf{g}= dr^2+Q(d\theta,d\theta) = dr_{1}^2 +dr_{2}^2 + Q_{ij}d \theta^{i} d\theta^{j},
\end{equation}
where the coefficients $Q_{ij}=Q_{ij}(r,\theta)$ are smooth functions.

As a demonstration of this framework in the simpler setting of rank 1, consider the hyperbolic plane $\mathbb{H}^2$. The half plane coordinates $(U,V)$, $U>0$ may be transformed to orthogonal horocyclic coordinates $(r,\theta)$ by $r=\log U$ and $\theta=V$ to find
\begin{equation}
\mathbf{g}_{-1}=\frac{dU^2 + dV^2}{U^2} = dr^2 + e^{-2r} d\theta^2.
\end{equation}
Here the flat $\mathcal{F}_{x_0}$ in the upper half plane model with $x_0 =(0,1)$ is the positive $U$-axis $\{V=0\}$, and the horocycle $N\cdot x_0$ is the horizontal line $\{U=1\}$.


For any unit tangent vector $Z \in T_x \mathbf{X}$ perpendicular to $\mathcal{F}_x$, the sectional curvature
\begin{equation}
\mathcal{K}(Z,\gamma_{x}'(0))=\langle R(Z, \gamma_{x}'(0))\gamma_{x}'(0), Z \rangle
\end{equation}
is negative, since $\mathcal{F}_x$ is a flat of maximal dimension. Moreover, such curvatures are uniformly negative (bounded away from zero) by compactness of the set of unit normal vectors to $\mathcal{F}_x$ and the homogeneity of $\mathbf{X}$.  The uniform (in $x$ as well as choice of 2-plane) upper and the lower bounds of these curvatures will be denoted by
\begin{equation}\label{curvaturebounds}
-c^2 \leq  \mathcal{K} \leq -b^{2} < 0.
\end{equation}

\begin{lemma}\label{lemmaimhof}
Let $J$ be a Jacobi field perpendicular to the
flat $\mathcal{F}_x$ along an arc-length parameterized geodesic $\gamma(s)\in\mathcal{F}_x$. Assume further that the Jacobi field is stable in that it is bounded as $s\rightarrow -\infty$, then
\begin{equation}
e^{bs}|J(0)| \leq|J(s)|\leq e^{cs}|J(0)|.
\end{equation}
\end{lemma}

\begin{proof}
This follows with slight modification from the proof of Theorem 2.4 in \cite{HeintzeImhof}, which relies on Proposition 4.1 in \cite{ImhofRuh}. The key observation is that the proof of Proposition 4.1 in \cite{ImhofRuh} does not
use the bounds on all sectional curvatures, but rather only those appearing in \eqref{curvaturebounds}.
\end{proof}

\begin{lemma}\label{lemmaQ}
For any vector $\xi\in\mathbb{R}^3$ and $i=1,2$
\begin{equation} \label{ineq:ab}
2b Q(\xi,\xi) \leq \partial_{r_i} Q(\xi,\xi) \leq 2c Q(\xi,\xi).
\end{equation}
\end{lemma}

\begin{proof}
Let $\psi:\mathbb{R}^5\rightarrow\mathbf{X}$ denote the global coordinate patch constructed above, so that $\psi^{-1}(x)=(r(x),\theta(x))$. Consider the geodesics
$\gamma_{\xi_0+\varepsilon\xi}: s\mapsto \psi(s,0 ,\xi_0+\varepsilon\xi)$ where $\varepsilon$ is a variation parameter and $\xi_0\in\mathbb{R}^3$ is fixed. If $v=(0,\xi)\in\mathbb{R}^5$ then $J_{\xi}=d\psi(v)$ is a Jacobi field along the geodesic $\gamma_{\xi_0}$. Moreover this Jacobi field is stable since
$d_{\mathbf{X}}(\gamma_{\xi_0+\varepsilon\xi}(s),\gamma_{\xi_0}(s))$ is bounded as $s\rightarrow -\infty$. Observe that
\begin{equation}
Q_{x}(\xi,\xi)=\mathbf{g}(d\psi_{\psi^{-1}(x)}(v),d\psi_{\psi^{-1}(x)}(v))
=|J_{\xi}(x)|^2,
\end{equation}
so the inequalities \eqref{ineq:ab} measure the logarithmic growth rate of stable Jacobi fields.

If $s \leq t$ then Lemma \ref{lemmaimhof} implies that
\begin{equation}
e^{2b(t-s)}|J_{\xi}(\gamma_{\xi_0}(s))|^2
\leq|J_{\xi}(\gamma_{\xi_0}(t))|^2\leq e^{2c(t-s)}|J_{\xi}(\gamma_{\xi_0}(s))|^2.
\end{equation}
The desired result now follows for $i=1$ by taking logarithms, dividing by $t -s$, and letting $t\rightarrow s$. Similar arguments hold for $i=2$.
\end{proof}

Consider a smooth map $\varphi:\mathbb{R}^3\setminus\Gamma\rightarrow \mathbf{X}$ with Dirichlet energy density
\begin{equation}
|d \varphi|^2 = | \nabla (r_1 \circ \varphi) |^2 + | \nabla (r_2 \circ \varphi) |^2 + Q  \Big( \nabla (\theta \circ \varphi), \nabla (\theta \circ \varphi) \Big),
\end{equation}
where the norms are computed with respect to the Euclidean metric $\delta$ in \eqref{flatmetric} and
\begin{equation}
Q \Big( \nabla (\theta \circ \varphi), \nabla (\theta \circ \varphi) \Big)
=Q_{ij}\left(\partial_{\rho}\theta^{i}\partial_{\rho}\theta^{j}
+\partial_{z}\theta^{i}\partial_{z}\theta^{j}\right).
\end{equation}
Let $\overline{\Omega}\subset\mathbb{R}^3\setminus\Gamma$ be the closure of a bounded domain situated away from the axis, and define the local Dirichlet energy
\begin{equation}
E_{\Omega}(\varphi)=\frac{1}{2}\int_{\Omega}|d\varphi|^2.
\end{equation}
Two of the harmonic map equations associated with the Dirichlet energy are
\begin{equation}\label{hmequations}
\Delta_\delta r_i = \partial_{r_i} Q(\nabla \theta, \nabla \theta),\quad\quad\quad i= 1,2.
\end{equation}
It then follows from Lemma \ref{lemmaQ} that each $r_i$ is subharmonic. Therefore if $\overline{\Omega}\subset\Omega'$ with $\overline{\Omega'}\subset\mathbb{R}^{3}\setminus\Gamma$ and $\chi\in C^{\infty}_{c}(\Omega')$ is a cut-off function with $\chi=1$ on $\Omega$, then multiplying by $\chi^2 r_i$ and integrating by parts produces
\begin{equation}\label{estimate1}
\int_{\Omega'} \chi^2 |\nabla r_i|^2 \leq 4 \left(\sup_{\Omega'} r_i^2\right)  \int_{\Omega'} | \nabla \chi |^2.
\end{equation}

Next combine \eqref{ineq:ab} with \eqref{hmequations} to obtain
\begin{equation}
\Delta_{\delta}r_{i}\geq 2b Q(\nabla\theta,\nabla\theta).
\end{equation}
Then multiplying by $\chi^2$, integrating by parts, and applying \eqref{estimate1} yields
\begin{equation}\label{estimate2}
\int_{\Omega'} \chi^2 Q(\nabla \theta, \nabla \theta) \leq \frac{1}{b} \int_{\Omega'} \chi \nabla \chi \cdot \nabla r \leq
\frac{2}{b} \left(\sup_{\Omega'} r_i\right)  \int_{\Omega'}| \nabla \chi |^2.
\end{equation}
Together \eqref{estimate1} and \eqref{estimate2} give the desired local energy estimate
\begin{equation}
E_{\Omega}(\varphi) \leq \Big[ 4 ( \sup_{\Omega'} r_1^2 +  \sup_{\Omega'} r_2^2 ) + \frac{2}{b} (\sup_{\Omega'} r_1 + \sup_{\Omega'} r_2 )  \Big] \int_{\Omega'}|\nabla \chi|^2.
\end{equation}

\begin{theorem}\label{energybound}
Let $\varphi:\mathbb{R}^3 \setminus\Gamma\rightarrow\mathbf{X}$ be a harmonic map and $\Omega\subset\mathbb{R}^3 \setminus\Gamma$ be a bounded domain. If $\varphi:\Omega\rightarrow B_{\mathcal{R}}(x_0)$ then
\begin{equation}
E_{\Omega}(\varphi)\leq C,
\end{equation}
where the constant $C$ depends only on the radius $\mathcal{R}$ of the geodesic ball and $\Omega$.
\end{theorem}



\section{Existence and Uniqueness} \label{existence}
\label{sec7} \setcounter{equation}{0}
\setcounter{section}{7}

In this section, we complete the proof of Theorem \ref{main} and prove the existence and uniqueness of a harmonic map $\varphi\colon\R^3\setminus\Gamma\to \mathbf{X}$ asymptotic to the model map $\varphi_0$ constructed in Section \ref{sec5}. Now that all the ingredients are in place, the proof is the same as in \cite{weinstein96}. Nevertheless, we include it here for the sake of completeness.
Let $\varepsilon>0$ and define $\Omega_\varepsilon=\{y\in\R^3\colon d_{\mathbb{R}^3}(y,\Gamma)>\varepsilon, \text{ } y\in B_{1/\varepsilon}(0)\}$. Since the target $\mathbf{X}$ is nonpositively curved, there is a smooth harmonic map $\varphi_\varepsilon\colon\Omega_\varepsilon\to\mathbf{X}$ such that $\varphi_\varepsilon=\varphi_0$ on $\partial\Omega_\varepsilon$. We quote the following lemma from \cite{weinstein96}, which essentially shows that the obstruction to a subharmonic distance function is given by the tension.

\begin{lemma}
Let $\varphi_1,\varphi_2\colon\Omega\to\mathbf{X}$ be smooth maps into a nonpositively curved target. Then
\begin{equation}
	\Delta\left( \sqrt{1 + d_{\mathbf{X}}(\varphi_1,\varphi_2)^2} \right) \geq
	-\left( |\tau(\varphi_1)| + |\tau(\varphi_2)| \right).
\end{equation}
\end{lemma}

Set $\varphi_1=\varphi_\varepsilon$ and $\varphi_2=\varphi_0$, and note that $\tau(\varphi_\varepsilon)=0$. The remaining tension may be estimated by $\Delta w\leq -|\tau(\varphi_0)|$, where $w>0$ and $w\rightarrow 0$ at infinity in $\mathbb{R}^3$.
This is possible due to the boundedness and decay of $|\tau(\varphi_0)|$ as given in Theorem \ref{model}. In particular we may take $w=c(1+r^2)^{-1/4}$ so that
\begin{equation}
\Delta w\leq -\frac{c}{4}(1+r^2)^{-5/4}\leq -|\tau(\varphi_0)|,
\end{equation}
if the constant $c>0$ is chosen sufficiently large. It follows that
\begin{equation}
	\Delta\left( \sqrt{1 + d_{\mathbf{X}}(\varphi_{\varepsilon},\varphi_0)^2} - w \right) \geq 0, \quad\quad\quad\quad
	\sqrt{1 + d_{\mathbf{X}}(\varphi_{\varepsilon},\varphi_0)^2} - w \leq 1 \text{ } \text{ on }\text{ }\partial\Omega_\varepsilon.
\end{equation}
The maximum principle then yields a uniform $L^\infty$ bound
\begin{equation}\label{c0bound}
\sqrt{1 + \dist_{\mathbf{X}}(\varphi_{\varepsilon},\varphi_0)^2} \leq 1 + w \text{ }\text{ }\text{ on }\text{ }\text{ } \Omega_\varepsilon.
\end{equation}
Fix a domain $\Omega$ such that $\overline{\Omega}\subset\R^3\setminus\Gamma$ and take $\varepsilon>0$  small enough to have $\overline{\Omega}\subset\Omega_\varepsilon$. The $L^\infty$ estimate combined with Theorem \ref{energybound} produces an energy bound on $\Omega$ independent of $\varepsilon$. Furthermore, consider the Bochner identity
\begin{equation}
\Delta|d\varphi_\varepsilon|^2 = |\hat{\nabla} d\varphi_\varepsilon|^2
	- \prescript{\mathbf{X}}{}{\operatorname{Riem}(d\varphi_\varepsilon,
d\varphi_\varepsilon,d\varphi_\varepsilon,d\varphi_\varepsilon)}. 
\end{equation}
Nonpositivity of the curvature shows that $|d\varphi_{\varepsilon}|^2$ is subharmonic. Thus a Moser iteration may be applied to find a uniform pointwise bound from the the energy estimate, namely
\begin{equation}
	\sup_{\Omega'} |d\varphi_\varepsilon|^2 \leq C \int_\Omega |d\varphi_\varepsilon|^2\leq C'
\end{equation}
where $\overline{\Omega}'\subset\Omega$. Finally, using the harmonic map equations combined with the pointwise gradient and $L^{\infty}$ bounds, we may now bootstrap to obtain uniform a priori estimates for all derivatives of $\varphi_\varepsilon$ on $\Omega'$.
By letting $\varepsilon\rightarrow 0$, it follows that there exists a subsequence which converges together with any number of derivatives on $\Omega'$. In the usual way, by choosing a sequence of exhausting domains and taking a diagonal subsequence, a sequence $\varphi_{\varepsilon_{i}}$ is produced which converges uniformly on compact subsets as $\varepsilon_i\rightarrow 0$. The limit $\varphi$ is smooth and harmonic, and satisfies the $L^\infty$ bound so that it is also asymptotic to $\varphi_0$.

The proof of uniqueness is straightforward. If $\varphi_1$ and $\varphi_2$ are two harmonic maps asymptotic to $\varphi_0$, then they are asymptotic to each other so that $d_{\mathbf{X}}(\varphi_1,\varphi_2)\leq C$.
Moreover
\begin{equation}
\Delta\left(\sqrt{1 + d_{\mathbf{X}}(\varphi_1,\varphi_2)^2} \right)\geq 0,
\end{equation}
and since the set $\Gamma$ on which $d_{\mathbf{X}}(\varphi_1,\varphi_2)$ may not be fully regular is of codimension 2, $\sqrt{1 + d_{\mathbf{X}}(\varphi_1,\varphi_2)^2}$ is weakly subharmonic and the maximum principle applies \cite{Weinstein}*{Lemma 8}.
As $\sqrt{1 + d_{\mathbf{X}}(\varphi_1,\varphi_2)^2}\to 1$ at infinity, it follows that $\sqrt{1 + d_{\mathbf{X}}(\varphi_1,\varphi_2)^2}\leq 1$.
Consequently $\varphi_1=\varphi_2$.


\subsection{Rod Data for the Harmonic Map}


Having constructed a harmonic map asymptotic to a prescribed model map, it remains to show that the rod data set arising from the harmonic map agrees with that of the model map. Let $\Phi=(F,\omega):\mathbb{R}^3\setminus\Gamma\rightarrow\tilde{\mathbf{X}}\cong SL(3,\mathbb{R})/SO(3)$ denote the characterization of the harmonic map in the space of symmetric positive definite matrices, and let $\Phi_0=(F_0,\omega_0)$ denote the model map asymptotic to $\Phi$. Recall that $F=(f_{ij})$ is a $2\times 2$ symmetric positive definite matrix on $\mathbb{R}^3\setminus\Gamma$ representing the fiber metric (associated with the rotational Killing directions) in a bi-axisymmetric stationary spacetime, and $\omega=(\omega_1,\omega_2)^t$ are the twist potentials. The rod data associated with $\Phi$ consists of the kernel of $F$ and the value of $\omega$ on the axis.

\begin{theorem}\label{rodstructharmonic}
If $\Phi$ is asymptotic to $\Phi_0$ then $\mathrm{ker} \text{ }F=\mathrm{ker}\text{ } F_0$ at each point of $\Gamma$, and $\omega=\omega_0$ on each axis rod. In particular, the two maps respect the same rod data set. Furthermore, if $\Phi$ is harmonic then $d_{\tilde{\mathbf{X}}}(\Phi_0,\Phi)\rightarrow 0$ at infinity in $\mathbb{R}^3$.
\end{theorem}

Before proving this result we record several observations. Since the metric on $\tilde{\mathbf{X}}$ is $G$-invariant, the distance function is preserved under the action of left translation
\begin{equation}
d_{\tilde{\mathbf{X}}}(\Phi_0,\Phi)=
d_{\tilde{\mathbf{X}}}(Id,\tilde{L}_{B^{-1}}\Phi),
\end{equation}
where $B\in SL(3,\mathbb{R})$ satisfies $BB^{t}=\Phi_0$. Note that
\begin{equation}
\tilde{L}_{B^{-1}}\Phi=B^{-1}
\Phi(B^{-1})^{t}=e^{W}
\end{equation}
for some symmetric $W$ with $\mathrm{Tr}\text{ }W=0$. Since the Riemannian exponential map and the matrix exponential coincide for $\tilde{\mathbf{X}}$, Hadamard's theorem applies (using the fact that $\tilde{\mathbf{X}}$ is complete, simply connected, with nonpositive curvature) to show that the exponential map is a diffeomorphism, and the geodesic $\gamma(t)= e^{tW}$ is minimizing. Therefore \eqref{identitymetric} yields
\begin{equation}
d_{\tilde{\mathbf{X}}}(Id,\tilde{L}_{B^{-1}}\Phi)=|\gamma'(0)|=|W|
=\sqrt{\mathrm{Tr}(WW^t)}
=\sqrt{\mathrm{Tr}(W^2)}.
\end{equation}

Now consider the function from the Mazur identity \cite{IdaIshibashiShiromizu}, namely
\begin{align}
\begin{split}
\mathrm{Tr}\left(\Phi_{0}^{-1}\Phi\right)=&
\mathrm{Tr}\left((B^{-1})^{t}B^{-1}
\Phi(B^{-1})^{t}B^{t}\right)\\
=&\mathrm{Tr}\left(B^{-1}\Phi(B^{-1})^t\right)\\
=&\mathrm{Tr}\text{ } e^{W}.
\end{split}
\end{align}
Since $e^{W}$ is symmetric and positive definite it may be diagonalized with positive eigenvalues
$\lambda_{i}$, $i=1,2,3$. We then have
\begin{equation}
\mathrm{Tr}\text{ }e^{W}=\lambda_{1}+\lambda_{2}+\lambda_{3},\quad\quad\quad
\mathrm{Tr}(W^2)=(\log\lambda_{1})^2+(\log\lambda_{2})^2+(\log\lambda_{3})^2,
\end{equation}
and since $W$ has zero trace
\begin{equation}\label{jfjfj}
\log\lambda_1+\log\lambda_2+\log\lambda_3=0.
\end{equation}
If $\mathrm{Tr}\text{ }e^{W}\leq c$ then it is not difficult to see that \eqref{jfjfj} implies $\mathrm{Tr}(W^2)\leq c_1$. Conversely if $\mathrm{Tr}(W^2)\leq c^2$
then each $|\log\lambda_{i}|\leq c$, and it holds that $\mathrm{Tr}\text{ }e^{W}\leq 3e^{c}$. We have thus proved the following.

\begin{lemma}\label{lem1}
$d_{\tilde{\mathbf{X}}}(\Phi_0,\Phi)$ is uniformly bounded if and only if
the Mazur quantity $\mathrm{Tr}\left(\Phi_{0}^{-1}\Phi\right)$
is uniformly bounded.
\end{lemma}

\noindent\textit{Proof of Theorem \ref{rodstructharmonic}.}
If $\Phi$ is asymptotic to $\Phi_0$ then $d_{\tilde{\mathbf{X}}}(\Phi_0,\Phi)\leq c_0$, that is the distance is uniformly bounded, in particular near $\Gamma$.
By Lemma \ref{lem1} this implies that the Mazur function is also uniformly bounded
\begin{equation}
\mathrm{Tr}\left(\Phi_{0}^{-1}\Phi\right)\leq c.
\end{equation}
Moreover this quantity may be computed in terms of $F$, $F_0$, $\omega$, and $\omega_0$ as
\begin{equation}\label{computationFF}
\mathrm{Tr}\left(\Phi_{0}^{-1}\Phi\right)
=\frac{f_0}{f}+\mathrm{Tr}(FF_{0}^{-1})+\frac{1}{f}(\omega-\omega_0)^{t}F_{0}^{-1}
(\omega-\omega_0),
\end{equation}
where $f=\det F$ and $f_0=\det F_0$. Since each of the terms on the right-hand side is nonnegative and the roles of $\Phi$ and $\Phi_0$ maybe reversed, we have
\begin{equation}\label{potential}
\frac{f_0}{f}\leq c,\quad\quad \frac{f}{f_0}\leq c,\quad\quad
\mathrm{Tr}(FF_{0}^{-1})\leq c,\quad \quad \frac{1}{f}(\omega-\omega_0)^{t}F_{0}^{-1}(\omega-\omega_{0})\leq c.
\end{equation}
It follows that
\begin{equation}
c^{-1}f_{0}\leq f\leq cf_{0}.
\end{equation}

Next, since $F_0$ is symmetric it may be diagonalized with an orthogonal matrix $O$ so that $F_0=ODO^t$ where $D=\mathrm{diag}(\mu_1, \mu_2)$. Working now at a point on an axis rod, the kernel of $F_0$ is 1-dimensional and so it may be assumed without loss of generality that $c^{-1}f_0\leq \mu_1\leq cf_0$ and $0<c^{-1}\leq \mu_2\leq c$. Let $\tilde{F}=O^t FO$ then
\begin{equation}
\mathrm{Tr}(FF_{0}^{-1})=\mathrm{Tr}(FO D^{-1}O^t)=\mathrm{Tr}(O O^t FO D^{-1}O^t)
=\mathrm{Tr}(\tilde{F}D^{-1})=\tilde{f}_{11}\mu_{1}^{-1}+\tilde{f}_{22}\mu_{2}^{-1}.
\end{equation}
Therefore
\begin{equation}
\tilde{f}_{11}\mu_{2}+\tilde{f}_{22}\mu_{1}\leq c\mu_1 \mu_2= cf_0,
\end{equation}
so that
\begin{equation}
\tilde{f}_{11}\leq cf_0,\quad\quad\tilde{f}_{22}\leq \frac{cf_0}{\mu_1}\leq c_1.
\end{equation}
Furthermore
\begin{equation}
f=\tilde{f}_{11}\tilde{f}_{22}-\tilde{f}_{12}^{2}\leq\tilde{f}_{11}\tilde{f}_{22}
\leq cf_{0}\tilde{f}_{22},
\end{equation}
which produces the lower bound
\begin{equation}
\tilde{f}_{22}\geq \frac{c^{-1}f}{f_0}\geq c_{1}^{-1}.
\end{equation}
In order to control the cross terms, observe that from the above
\begin{equation}
\tilde{f}_{12}^{2}=\tilde{f}_{11}\tilde{f}_{22}+f\leq cf_0.
\end{equation}
In conclusion we obtain
\begin{equation}
\tilde{f}_{11}\leq cf_0,\quad\quad |\tilde{f}_{12}|\leq c\sqrt{f_0},\quad\quad
c^{-1}\leq \tilde{f}_{22}\leq c.
\end{equation}
Therefore, on an axis rod both $D=O^t F_0 O$ and $\tilde{F}=O^t FO$ have the same kernel, and thus $F_0$ and $F$ have the same kernel. Similar arguments hold for a horizon rod.

Let us now show that the potentials agree on an axis rod. From \eqref{potential} it holds that
\begin{equation}
(\tilde{\omega}-\tilde{\omega}_0)^{t}D^{-1}(\tilde{\omega}-\tilde{\omega}_0)
=(\omega-\omega_0)^{t}F_{0}^{-1}(\omega-\omega_0)\leq cf\leq c_1 f_0,
\end{equation}
where
\begin{equation}
(\tilde{\omega}-\tilde{\omega}_0)=O^t(\omega-\omega_0).
\end{equation}
It follows that
\begin{equation}
\mu_{1}^{-1}(\tilde{\omega}^1-\tilde{\omega}_{0}^1)^2
+\mu_{2}^{-1}(\tilde{\omega}^2-\tilde{\omega}_{0}^2)^2\leq cf_0,
\end{equation}
which implies
\begin{equation}
f_{0}^{-1}(\tilde{\omega}^1-\tilde{\omega}_{0}^1)^2
+(\tilde{\omega}^2-\tilde{\omega}_{0}^2)^2\leq c_1 f_0.
\end{equation}
We then have
\begin{equation}
|\tilde{\omega}-\tilde{\omega}_{0}|^2\leq cf_0\quad\quad \Rightarrow
\quad\quad |\omega-\omega_0|^2\leq cf_0,
\end{equation}
showing that $\omega=\omega_0$ on an axis rod.

Lastly if $\Phi$ is harmonic then according to the $L^\infty$ bound \eqref{c0bound}, which holds globally for $\varphi$ in place of $\varphi_{\varepsilon}$, it must hold that
$d_{\tilde{\mathbf{X}}}(\Phi_0,\Phi)\rightarrow 0$ at infinity in $\mathbb{R}^3$ since
$w\rightarrow 0$ in this limit. \qed

\section{Reconstruction of the Spacetime Metric} \label{reconstruct}
\label{sec8} \setcounter{equation}{0}
\setcounter{section}{8}



Let $\Phi=(F,\omega):\mathbb{R}^3\setminus\Gamma\rightarrow\tilde{\mathbf{X}}\cong
SL(3,\mathbb{R})/SO(3)$ be the harmonic map constructed in the previous section, represented in the space of symmetric positive definite matrices. Here we show how to build an asymptotically flat bi-axisymmetric stationary vacuum spacetime, which inherits the prescribed rod data set associated with the harmonic map. Note that the functions $F=(f_{ij})$ and $\omega=(\omega_1,\omega_2)^t$ comprising the harmonic map are defined and smooth on the right-half plane $\{(\rho,z)\mid \rho>0\}$, which will serve as the orbit space for the spacetime. The spacetime metric is given by \eqref{metric}, and it suffices to show how each coefficient in \eqref{metric} arises from $\Phi$. The resulting spacetime will be asymptotically flat in light of the
decay of the model map $\Phi_0$ and the fact that, by
Theorem \ref{rodstructharmonic}, $d_{\tilde{\mathbf{X}}}(\Phi_0,\Phi)\rightarrow 0$ at infinity in $\mathbb{R}^3$.

First observe that $\sigma$ is immediately obtained from \eqref{sigma}, since the orbit space is simply connected and the form on the right-hand side is closed as a result of the harmonic map equations. It remains to find $A^{(i)}=v^{i}dt$, which will be derived from the harmonic map components $\omega_{i}$. By solving for $dA^{(i)}$ in \eqref{chi} we get
\begin{equation}\label{0}
d A^{(i)}=-\frac{1}{2}f^{-1}f^{ij}\star_{3} d\omega_{j}.
\end{equation}
Observe that from Cartan's magic formula and the fact that $\partial_{t}$ is a Killing field
\begin{equation}
\iota_{\partial_{t}} d A^{(i)}=-d \iota_{\partial_{t}}A^{(i)}=-dv^{i}.
\end{equation}
It follows that if
\begin{equation}\label{1}
\iota_{\partial_{t}}\left(f^{-1}f^{ij}\star_{3} d\omega_{j}\right)
\end{equation}
is closed, then we may find $v^{i}$ by quadrature from the equation
\begin{equation}
dv^{i}=\frac{1}{2}\iota_{\partial_{t}}\left(f^{-1}f^{ij}\star_{3} d\omega_{j}\right).
\end{equation}
It turns out that showing \eqref{1} is closed is equivalent to parts of the harmonic map equations. To see this, let $\epsilon_{3}$ denote the volume form of $g_{3}$. Then
\begin{equation}
(\star_{3} d\omega_{j})^{ab}=\epsilon_{3}^{abc}\partial_{c}\omega_{j},
\end{equation}
and
\begin{align}\label{222}
\begin{split}
\iota_{\partial_{t}}\star_{3}d\omega_{j}
=&\epsilon_{3}(\partial_{t},\partial_{\rho},\partial_{c})
\partial^{c}\omega_{j} d\rho
+\epsilon_{3}(\partial_{t},\partial_{z},\partial_{c})\partial^{c}\omega_{j} dz\\
=&\rho e^{2\sigma}\partial^{z}\omega_{j} d\rho
-\rho e^{2\sigma}\partial^{\rho}\omega_{j} dz\\
=&\rho\partial_{z}\omega_{j} d\rho-\rho\partial_{\rho}\omega_{j} dz.
\end{split}
\end{align}
Therefore
\begin{align}
\begin{split}
d\left(f^{-1}f^{ij}\iota_{\partial_{t}}\star_{3} d\omega_{j}\right)=&
d\left(\rho f^{-1}f^{ij}\partial_{z}\omega_{j} d\rho
-\rho f^{-1}f^{ij}\partial_{\rho}\omega_{j} dz\right)\\
=&\left[\partial_{z}(\rho f^{-1}f^{ij}\partial_{z}\omega_{j})+
\partial_{\rho}(\rho f^{-1}f^{ij}\partial_{\rho}\omega_{j}) \right]dz\wedge d\rho\\
=&\operatorname{div}_{\mathbb{R}^3}\left(f^{-1} f^{ij}\nabla\omega_{j}\right)dz\wedge d\rho\\
=&0,
\end{split}
\end{align}
where the last equality arises from the second set of harmonic maps equations
in \eqref{eulerlagrange}.
Another way to obtain this calculation is to observe that
\begin{equation}\label{2}
f^{-1}f^{ij}\iota_{\partial_{t}}\star_{3} d\omega_{j}
=\ast\left(f^{-1} f^{ij} d\omega_{j}\right)
\end{equation}
and $\operatorname{div}_{\mathbb{R}^3}=\ast d \ast$,
where $\ast$ is the Hodge star operator with respect to $\delta$ on $\mathbb{R}^3$.
Lastly, it is clear from the equations involved that $\sigma$ and $v^i$ are bi-axisymmetric.

\subsection{Regularity} \label{conical}


The metric reconstructed above from a solution of the harmonic map problem is defined on $\R\times(\R^3\setminus\Gamma)\times U(1)$. In order to extend this metric across $\Gamma$, two steps must be completed as described below.

\subsubsection{Analytic regularity}

The metric coefficients in~\eqref{metric} must be shown to be smooth and even in $\rho$ up to $\Gamma$. This was achieved in the 4D case in~\cite{weinstein90}, and then extended to the non-axially symmetric case in~\cite{litian}. We believe that these methods are applicable to the 5D setting as well.

\subsubsection{Conical singularities}

In addition to the analytic regularity mentioned above, conical singularities on axis rods must be ruled out.  A conical singularity at a point on an axis rod $\Gamma_{l}$ is measured by the angle deficiency $\theta\in(-\infty,2\pi)$ given by
\begin{equation}
\frac{2\pi}{2\pi-\theta}=\lim_{\rho\rightarrow 0}\frac{2\pi\cdot\mathrm{Radius}}
{\mathrm{Circumference}}=\lim_{\rho\rightarrow 0}
\frac{\int_{0}^{\rho}\sqrt{f^{-1}e^{2\sigma}}}
{\sqrt{f_{ij}u^{i}u^{j}}}=\lim_{\rho\rightarrow 0}
\sqrt{\frac{\rho^{2}f^{-1}e^{2\sigma}}{f_{ij}u^{i}u^{j}}},
\end{equation}
where $u=(u^1,u^2)=(m_{l},n_{l})$ is the associated rod structure so that $u$ is in the kernel of $F$ at $\rho=0$. Absence of a conical singularity is characterized by a zero angle deficiency, that is when the right-hand side is 1; this is referred to as the balancing condition in Section \ref{sec1}. By a standard change of coordinates from polar to Cartesian, it is straightforward to check that once analytic regularity has been established this condition is necessary and sufficient for the metric to be extendable across the axis.

Let us denote by $b_l$ the value of $\log\left(\frac{2\pi}{2\pi-\theta}\right)$ on the axis rod $\Gamma_l$. Then, similarly to the 4D case, it can be shown from~\eqref{sigma} that $b_l$ is constant on $\Gamma_l$. Moreover asymptotic flatness implies that $b_l=0$ on the two semi-infinite axis rods, $l=1,L+1$. Thus it it remains to investigate the value of $b_l$ on the bounded axis rods. In the example from Figure~\ref{domain}, to show regularity would only require showing that $b_3=0$ so that the angle deficit vanishes on the disk rod, between points $S$ and $C$.

In 4D very few cases have been worked out, see~\cites{weinstein94,litian91}. In the current 5D setting, it is known that some configurations without any conical singularity do exist as mentioned in the introduction. We conjecture that many more such regular solutions can be found. These questions will be investigated in a future paper.

\appendix

\section{Topology of Corners }
\label{appendix} \setcounter{equation}{0}

\begin{proposition}
In a stationary bi-axisymmetric spacetime, consider a corner defined by a top rod of structure
$(m, n)$ and a bottom rod of structure $(r, s)$, with $\mathrm{gcd}(m, n) = \mathrm{gcd}(r, s) = 1$. If
\begin{equation}
\det \left(
    \begin{array}{cc}
      m & n \\
      r  & s
    \end{array}
  \right)
= \pm 1,
\end{equation}
then the spacetime is locally diffeomorphic to $\mathbb{R}^5$ near the corner.
\end{proposition}

\begin{proof}
Let
\begin{equation}
V = m \partial_{\phi^1} + n \partial_{\phi^2}, \quad\quad\quad W = r \partial_{\phi^1} + s \partial_{\phi^2},
\end{equation}
be the Killing fields which vanish at the top rod and bottom rod, respectively.
The first goal is to show that there exists a change of variables $(\overline{\phi}^1, \overline{\phi}^2)$, which are also $2\pi$-periodic, such that the Killing fields take the form
\begin{equation}
V =  \partial_{\overline{\phi}^1} + \nu \partial_{\overline{\phi}^2},\quad\quad\quad
W =  \partial_{\overline{\phi}^2},
\end{equation}
for some integer $\nu$. The coordinate transformation may be realized by a $2 \times 2$ matrix having integer entries
\begin{equation}
A = \left(
    \begin{array}{cc}
      a & b \\
      c  & d
    \end{array}
  \right),
\end{equation}
with $\det A = -1$. Namely
\begin{align}
\begin{split}
\overline{\phi}^1 = & a \phi^1 + b \phi^2,\quad\quad\quad\quad \quad
\phi^1 =  -d  \overline{\phi}^1+ b \overline{\phi}^2,\\
\overline{\phi}^2 = & c \phi^1 + d \phi^2,\quad\quad\quad\quad\quad
\phi^2 = c \overline{\phi}^1  -a \overline{\phi}^2.
\end{split}
\end{align}

To see that the new variables are $2\pi$-periodic consider the translation
$\overline{\phi}^1 \mapsto \overline{\phi}^1 + 2 \pi$, which corresponds to
\begin{equation}
\phi^1 \mapsto \phi^1 - 2 \pi d \quad\text{ }\text{ and  }\text{ }\quad \phi^2 \mapsto \phi^2 + 2 \pi c.
\end{equation}
Since $c,d\in\mathbb{Z}$ and $\phi^i$ are $2\pi$-periodic, it follows that $\overline{\phi}^1$ has a period less than or equal to $2 \pi$. If the period is $2  \alpha \pi$ for $0 < \alpha < 1$,   then the translation $\overline{\phi}^1 \mapsto \overline{\phi}^1 + 2 \pi\alpha$ would map to the same points, and as a consequence the shifts
\begin{equation}
\phi^1 \mapsto \phi^1 - 2 \pi \alpha d \quad\text{ }\text{ and  }\text{ }\quad \phi^2 \mapsto \phi^2 + 2 \pi\alpha c
\end{equation}
would give the identity map. This implies that $\alpha  d$ and $ \alpha  c $ are integers, which is impossible since
\begin{equation}
1 > \alpha = \alpha  |\det A| = | a (\alpha d) - b (\alpha c)| \neq 0.
\end{equation}
Similar arguments show that $\overline{\phi}^2$ has period $2 \pi$.

To find the matrix $A$ observe that
\begin{equation}
V  = (ma+nb)  \partial_{\overline{\phi}^1} + (mc+nd) \partial_{\overline{\phi}^2},\quad\quad
 W  = (ra+sb)  \partial_{\overline{\phi}^1} + (rc+sd) \partial_{\overline{\phi}^2}.
\end{equation}
Thus we aim to solve
\begin{equation}
\left(
    \begin{array}{c}
      1  \\
      \nu
    \end{array}
  \right)
  =
   A \left(
    \begin{array}{c}
      m  \\
      n
    \end{array}
  \right)
  =
   \left(
    \begin{array}{c}
      ma+nb  \\
      mc+nd
    \end{array}
  \right)
\end{equation}
and
\begin{equation}
\left(
    \begin{array}{c}
      0  \\
      1
    \end{array}
  \right)
  =
   A \left(
    \begin{array}{c}
      r  \\
      s
    \end{array}
  \right)
  =
   \left(
    \begin{array}{c}
      ra+sb  \\
      rc+sd
    \end{array}
  \right).
\end{equation}
Consider the second pair of equations
\begin{align}
\begin{split}
ra+sb = & 0, \\
rc + sd = & 1.
\end{split}
\end{align}
Choose $a = -s$ and $b = r$ to solve the first equation. The integers $c$ and $d$ may be found using Bezout's Lemma (Lemma \ref{Bezout}), which gives a solution satisfying $|c| \leq |s|$ and $|d| \leq |r|$, resulting in
\begin{equation}
A = \left(
    \begin{array}{cc}
      -s & r \\
      c & d
    \end{array}
  \right)
\end{equation}
with $\det A = -sd - cr =-1$.

We now have integers $\mu$ and $\nu$ defined by
\begin{equation}
   A \left(
    \begin{array}{c}
      m  \\
      n
    \end{array}
  \right)
  =
   \left(
    \begin{array}{c}
      \mu  \\
      \nu
    \end{array}
  \right).
\end{equation}
It turns out that $\mu=1$. To see this note that
\begin{equation}
A  \left(
    \begin{array}{cc}
      m & r \\
      n & s
    \end{array}
  \right) =
   \left(
    \begin{array}{cc}
      \mu & 0 \\
      \nu & 1
    \end{array}
  \right),
\end{equation}
so $\det A = -1$ together with the hypothesis of this proposition produces
\begin{equation}
\mu = \det A \det \left(
    \begin{array}{cc}
      m & r \\
      n & s
    \end{array}
  \right)  = \mp 1.
\end{equation}
If $\mu = -1$, simply choose
\begin{equation}
A = \left(
    \begin{array}{cc}
      s & -r \\
      c & d
    \end{array}
  \right)
\end{equation}
to achieve $\mu = 1$ if necessary.

In the new coordinate system the corner is defined by the rod structures
$(1, \nu)$ and $(0, 1)$. Then as described in Section \ref{sec4}, any simple
curve in the 2-dimensional orbit space which encircles the corner and connects the top rod to the bottom rod represents a lens space $L(1,\nu)\cong S^3$.
Therefore by foliating a neighborhood of the corner in the orbit space by such curves, we find that a punctured neighborhood of the corner in a time slice has topology $\mathbb{R}\times S^3 \cong\mathbb{R}^4 \setminus\{0\}$. It follows that there is a spacetime neighborhood of the corner which is diffeomorphic to $\mathbb{R}^5$.
\end{proof}

\end{document}